\documentclass[a4paper,11pt]{article}

\usepackage{amsmath}
\usepackage{amssymb}
\usepackage{amsthm}
\usepackage{graphicx}
\usepackage{enumerate}
\usepackage{cite}
\usepackage{fullpage}
\usepackage[usenames]{color}
\usepackage[nodayofweek]{datetime}
\usepackage{hyperref}
\usepackage{scalerel}
\usepackage{multirow}
\usepackage[mathlines]{lineno}

\allowdisplaybreaks

\newcommand{\Real}{\ensuremath{\mathbb{R}}}
\newcommand{\Plane}{\ensuremath{\mathbb{R}^2}}

\newcommand{\bd}{\ensuremath{\partial}}

\newcommand{\intr}{\ensuremath{\mathrm{int}}}

\newcommand{\seg}{\overline}
\newcommand{\width}{\mathrm{width}}
\newcommand{\area}{\mathrm{area}}
\newcommand{\peri}{\mathrm{peri}}

\newsavebox\Tbox
\savebox\Tbox{%\hbox{%
  \unitlength=0.1ex%
  \begin{picture}(8,10)
    \linethickness{0.1mm}
    \multiput(0,0)(0,2){4}{\line(0,1){1}}
    \multiput(0,0)(2,0){4}{\line(1,0){1}}
    \multiput(7,7)(0,-2){4}{\line(0,-1){1}}
    \linethickness{0.2mm}
    \put(7,7){\line(-1,0){7}}
    \put(3.5,7){\circle*{3.5}}
  \end{picture}}
\newsavebox\Bbox
\savebox\Bbox{%\hbox{%
  \unitlength=0.1ex%
  \begin{picture}(8,10)
    \linethickness{0.1mm}
    \multiput(7,7)(0,-2){4}{\line(0,-1){1}}
    \multiput(0,0)(0,2){4}{\line(0,1){1}}
    \multiput(7,7)(-2,0){4}{\line(-1,0){1}}
    \linethickness{0.2mm}
    \put(0,0){\line(1,0){7}}
    \put(3.5,0){\circle*{3.5}}
  \end{picture}}
\newsavebox\Lbox
\savebox\Lbox{%\hbox{%
  \unitlength=0.1ex%
  \begin{picture}(8,10)
    \linethickness{0.1mm}
    \multiput(0,0)(2,0){4}{\line(1,0){1}}
    \multiput(7,7)(0,-2){4}{\line(0,-1){1}}
    \multiput(7,7)(-2,0){4}{\line(-1,0){1}}
    \linethickness{0.2mm}
    \put(0,0){\line(0,1){7}}
    \put(0,3.5){\circle*{3.5}}
  \end{picture}}
\newsavebox\Rbox
\savebox\Rbox{%\hbox{%
  \unitlength=0.1ex%
  \begin{picture}(8,10)
    \linethickness{0.1mm}
    \multiput(0,0)(0,2){4}{\line(0,1){1}}
    \multiput(0,0)(2,0){4}{\line(1,0){1}}
    \multiput(7,7)(-2,0){4}{\line(-1,0){1}}
    \linethickness{0.2mm}
    \put(7,7){\line(0,-1){7}}
    \put(7,3.5){\circle*{3.5}}
  \end{picture}}
\newsavebox\Abox
\savebox\Abox{%\hbox{%
  \unitlength=0.1ex%
  \begin{picture}(8,10)
    \linethickness{0.1mm}
    \multiput(0,0)(0,2){4}{\line(0,1){1}}
    \multiput(0,0)(2,0){4}{\line(1,0){1}}
    \multiput(7,7)(0,-2){4}{\line(0,-1){1}}
    \multiput(7,7)(-2,0){4}{\line(-1,0){1}}
  \end{picture}}

\DeclareRobustCommand\mt{{\scalerel*{{\mathord{\usebox{\Tbox}}}}{b}}}
\DeclareRobustCommand\mb{{\scalerel*{{\mathord{\usebox{\Bbox}}}}{b}}}
\DeclareRobustCommand\ml{{\scalerel*{{\mathord{\usebox{\Lbox}}}}{b}}}
\DeclareRobustCommand\mr{{\scalerel*{{\mathord{\usebox{\Rbox}}}}{b}}}
\DeclareRobustCommand\ma{{\scalerel*{{\mathord{\usebox{\Abox}}}}{b}}}
\newcommand{\clend}{\vdash}
\newcommand{\crend}{\dashv}
\newcommand{\cmid}{+}

\newcommand{\extr}{q}

\newtheoremstyle{mytheorem}{3pt}{3pt}{\slshape}{}{\bfseries}{}{.5em}{}
\theoremstyle{mytheorem}
\newtheorem{lemma}{Lemma}
\newtheorem{theorem}{Theorem}

\newtheorem{corollary}{Corollary}
\newtheorem{observation}{Observation}

\theoremstyle{definition}

\makeatletter
\newbox\ProofSym
\setbox\ProofSym=\hbox{%
\unitlength=0.18ex%
\begin{picture}(10,10)
\put(0,0){\framebox(9,9){}} \put(0,3){\framebox(6,6){}}
\end{picture}}
\renewenvironment{proof}[1][Proof.]{\O@proof{#1}}{\O@endproof}
\def\O@proof#1{\trivlist
   \@topsep\z@\@topsepadd\smallskipamount%
   \@ifstar{\item[]}{\item[\hskip\labelsep\it #1 ]}}
\def\O@endproof{\hfill\copy\ProofSym\linebreak[3mm]\endtrivlist}
\makeatother

\def\denseitems{
    \itemsep1pt plus1pt minus1pt
    \parsep0pt plus0pt
    \parskip0pt\topsep0pt}
\newcommand*\patchAmsMathEnvironmentForLineno[1]{%
  \expandafter\let\csname old#1\expandafter\endcsname\csname #1\endcsname
  \expandafter\let\csname oldend#1\expandafter\endcsname\csname end#1\endcsname
  \renewenvironment{#1}%
     {\linenomath\csname old#1\endcsname}%
     {\csname oldend#1\endcsname\endlinenomath}}%
\newcommand*\patchBothAmsMathEnvironmentsForLineno[1]{%
  \patchAmsMathEnvironmentForLineno{#1}%
  \patchAmsMathEnvironmentForLineno{#1*}}%
\AtBeginDocument{%
\patchBothAmsMathEnvironmentsForLineno{equation}%
\patchBothAmsMathEnvironmentsForLineno{align}%
\patchBothAmsMathEnvironmentsForLineno{flalign}%
\patchBothAmsMathEnvironmentsForLineno{alignat}%
\patchBothAmsMathEnvironmentsForLineno{gather}%
\patchBothAmsMathEnvironmentsForLineno{multline}%
}
\begin{document}

%\baselineskip=14.0pt
%\linenumbers

\title{On the Minimum-Area Rectangular and Square Annulus Problem%
\thanks{%
This work was supported by Kyonggi University Research Grant 2018.
}
}

\author{%
Sang Won Bae\footnote{%
Division of Computer Science and Engineering, Kyonggi University, Suwon, Korea.
Email: \texttt{swbae@kgu.ac.kr} }
}

\date{%
\today\quad\currenttime
}

%
%
%\let\oldalign\align
%\let\oldendalign\endalign
%\renewenvironment{align}
%  {\linenomathNonumbers\oldalign}
%  {\oldendalign\endlinenomath}

%%%%%%%%%%%%%%
\maketitle

%\vspace{-0.3in}
\begin{abstract}
In this paper, we address the minimum-area rectangular and square annulus problem,
which asks a rectangular or square annulus of minimum area,
either in a fixed orientation or over all orientations,
that encloses
a set $P$ of $n$ input points in the plane.
To our best knowledge, no nontrivial results on the problem
have been discussed in the literature,
while its minimum-width variants have been intensively studied.
For a fixed orientation, we show reductions to well-studied problems:
the minimum-width square annulus problem and
the largest empty rectangle problem,
yielding algorithms of time complexity $O(n\log^2 n)$ and $O(n\log n)$
for the rectangular and square cases, respectively.
In arbitrary orientation,
we present $O(n^3)$-time algorithms for
the rectangular and square annulus problem by enumerating all
maximal empty rectangles over all orientations.
The same approach is shown to apply also to the minimum-width square annulus problem
and the largest empty square problem over all orientations,
resulting in $O(n^3)$-time algorithms for both problems.
Consequently, we improve
the previously best algorithm for the minimum-width square annulus problem
by a factor of logarithm,
and present the first algorithm for the largest empty square problem
in arbitrary orientation.
We also consider bicriteria optimization variants,
computing a minimum-width minimum-area or minimum-area minimum-width annulus.
\end{abstract}

%%%%%%%%%%%%%%%%%%%%%%%%%%%%%%%%%%%%%%
\section{Introduction} \label{sec:intro}
%%%%%%%%%%%%%%%%%%%%%%%%%%%%%%%%%%%%%%%

An annulus informally depicts a ring-shaped region in the plane,
often described by two concentric circles.
One can consider a generalization to any convex shape,
such as squares and rectangles.
Recently, the \emph{minimum-width annulus problem} has been studied intensively
by researchers,
in which a set $P$ of $n$ points in the plane is given
and one wants to find an annulus of a certain shape
with minimum width that encloses the set $P$ of points.
This problem can be seen as a typical geometric covering problem
that seeks a minimum-size geometric shape that covers a given set $P$ of points.
The annulus problem has applications in shape recognition, facility location, and curve fitting.

Among other shapes, the circular annulus problem has been most intensively studied
with an application to the roundness problem~\cite{w-nmppp-86,efnn-ravd-89,rz-epccmrsare-92}.
The first sub-quadratic $O(n^{\frac{8}{5}+\epsilon})$-time algorithm
was presented by Agarwal et al.~\cite{ast-apsgo-94}
The currently best exact algorithm takes $O(n^{\frac{3}{2}+\epsilon})$ time by Agarwal and Sharir~\cite{as-erasgop-96}.
Linear-time approximation schemes are also known
by Agarwal et al.~\cite{ahv-aemp-04} and by Chan~\cite{c-adwsecmwa-02}.
Abellanas et al.~\cite{ahimpr-bfr-03} considered minimum-width rectangular annuli
that are axis-parallel, and presented two algorithms taking $O(n)$ or $O(n \log n)$ time:
one minimizes the width over rectangular annuli with arbitrary aspect ratio
and the other does over  rectangular annuli with a prescribed aspect ratio, respectively.
Gluchshenko et al.~\cite{ght-oafepramw-09} presented an $O(n \log n)$-time algorithm
that computes a minimum-width axis-parallel square annulus,
and proved a matching lower bound,
while the second algorithm by Abellanas et al.\@ can do the same in the same time bound.
If one considers rectangular or square annuli in arbitrary orientation,
the problem becomes more difficult.
Mukherjee et al.~\cite{mmkd-mwra-13} presented an $O(n^2 \log n)$-time algorithm
that computes a minimum-width rectangular annulus in arbitrary orientation
and arbitrary aspect ratio.
The author~\cite{b-cmwsaao-18} recently showed that
a minimum-width square annulus in arbitrary orientation can be computed
in $O(n^3 \log n)$ time.
Other variants such as the minimum-width annulus problem with outliers~\cite{b-cmwsrao-19, aabckosy-mwao:csrc-19} and the maximum-width empty annulus problem~\cite{bbm-mwesra-19, dhmrs-leap-03}
have also been studied.

Like other geometric covering problems, one may consider several different objective functions to optimize for annuli.
Among those objectives, in this paper,
we address the problem of finding an annulus of \emph{minimum area}
enclosing the input points $P$, namely,
the \emph{minimum-area annulus problem}.
In the literature, it is surprisingly hard to find results on annuli of minimum area, but a few remarks about the minimum-area circular annulus problem.
As earlier work pointed out and Chan~\cite{c-adwsecmwa-02} discussed,
the minimum-area circular annulus problem can be formulated into
a linear programming, so can be solved in linear $O(n)$ time.
It should be mentioned that the currently best known algorithm for
the \emph{minimum-width} circular annulus problem
takes $O(n^{\frac{3}{2}+\epsilon})$ time~\cite{as-erasgop-96}.
Little is known, however, about rectangular and square annuli of minimum area.
The purpose of this paper is thus to gain the understanding of
the minimum-area rectangular and square annulus problem
and their relations with other well-known geometric problems,
and finally to achieve efficient algorithms.

\begin{figure}[tb]
\begin{center}
\includegraphics[width=.7\textwidth]{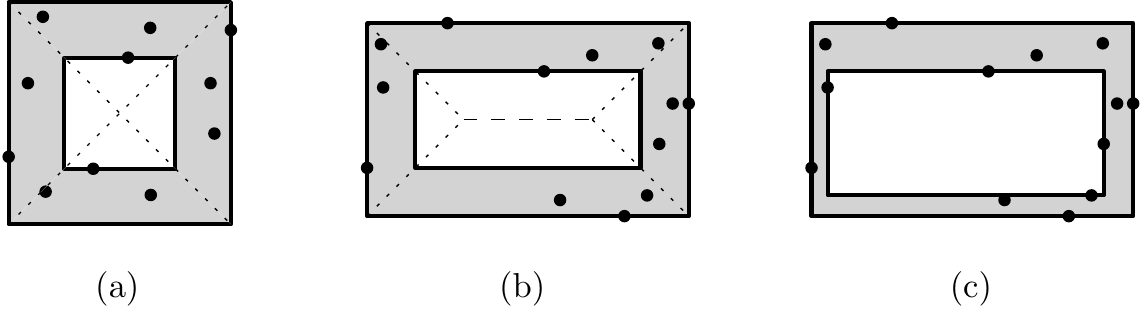}
\end{center}
%\vspace{-5mm}
\caption{(a) A square annulus, (b) a uniform rectangular annulus, and
 (c) a rectangular annulus that is not uniform, enclosing a set of points,
 all of which are axis-parallel.
 Each of the three here is of minimum possible width among all
 square, uniform rectangular, and rectangular annuli, respectively,
 enclosing the points.
 }
% \vspace{-5mm}
\label{fig:annulus}
\end{figure}

A circular or square annulus is defined by two concentric circles or squares,
respectively.
However, a rectangular annulus can be defined in several ways,
and, in the literature, two different definitions are often introduced.
In either case, a rectangular annulus is the closed region between
two side-parallel rectangles $R$ and $R'$ such that $R' \subseteq R$.
One may restrict the inner rectangle $R'$ to be an offset of the outer rectangle $R$
so that the distance from any point in the boundary of $R'$ to the boundary of $R$
is all uniform~\cite{ahimpr-bfr-03, b-cmwsrao-19},
while he/she may allow any rectangle contained in $R$
to be the inner rectangle $R'$~\cite{mmkd-mwra-13}.
In this paper, we call the latter a rectangular annulus in general
and the former a \emph{uniform} rectangular annulus.
See \figurename~\ref{fig:annulus} for an illustration.
It is not difficult to see that these two definitions of rectangular annuli
are equivalent for the \emph{minimum-width} problem and other variants,
while this is not the case for the \emph{minimum-area} problem.

In this paper, we present first efficient algorithms for
the problem of computing a minimum-area square annulus
or rectangular annulus, either in a fixed or over all orientations.
For the purpose, we show some relations between the minimum-area problem
and other well-known geometric optimization problems,
such as the \emph{minimum-width annulus} problem and
the \emph{largest empty rectangle} problem.
As a result, the minimum-area problem for square and rectangular annuli
is as hard as the minimum-width variant.
It is interesting to note that it is not the case for circular annuli.
We also consider the bicriteria optimization problems:
computing a minimum-area minimum-width annulus or a minimum-width minimum-area annulus.

\begin{table}[tb]
\caption{Summary on the currently best algorithms for the square or rectangular annulus problem.
Each column of the table marked by `W', `A', `AW', and `WA'
stands for `minimum-width', `minimum-area', `minimum-area minimum-width', and
`minimum-width minimum-area', respectively.
In each row, `f', `a', `S', `uR', and `R' are abbreviations for
`square annulus', `uniform rectangular annulus', `rectangular annulus',
`in a fixed orientation', and `in arbitrary orientation', respectively.
For each case, the table indicates the wort-case time complexity of
the currently best algorithm, omitting the big-Oh $O(\cdot)$ symbol,
and its reference.
Details on the parameters $r = O(n^2)$ and $t=O(n^2)$ can be found in
Theorems~\ref{thm:WA-R-f} and~\ref{thm:AW-R-a}, respectively.
}
\label{tbl:results}
\small\begin{center}
\begin{tabular}{ll||ll|ll|ll|ll}
                                         &           & \multicolumn{2}{l|}{W} & \multicolumn{2}{l|}{A} & \multicolumn{2}{l|}{AW} & \multicolumn{2}{l}{WA} \\ \hline \hline
\multicolumn{1}{l|}{\multirow{3}{*}{f}}  & S  & $n \log n$     & \cite{ahimpr-bfr-03,ght-oafepramw-09}   & $n \log n$     & Thm.\ref{thm:A-S-f}  & $n \log n$ & Coro.\ref{coro:AW/WA-S-f}  & $n \log n$  & Coro.\ref{coro:AW/WA-S-f}  \\
\multicolumn{1}{l|}{}                    & uR  & $n$ & \cite{ahimpr-bfr-03}  & $n$  & Thm.\ref{thm:A-uR-f}  & $n$  & Coro.\ref{coro:AW/WA-uR-f} & $n$              & Coro.\ref{coro:AW/WA-uR-f}   \\
\multicolumn{1}{l|}{}                    & R  & $n$     & \cite{mmkd-mwra-13}   & $n\log^2 n$    & Thm.\ref{thm:A-R-f}   & $n \log n$   & Thm.\ref{thm:AW-R-f} & $n\log n + r$  & Thm.\ref{thm:WA-R-f}  \\ \hline
\multicolumn{1}{l|}{\multirow{3}{*}{a}}  & S & $n^3$   & Thm.\ref{thm:W-S-a}   & $n^3$  & Thm.\ref{thm:A-S-a} & $n^3$  & Coro.\ref{coro:WA/AW-S-a} & $n^3$            & Coro.\ref{coro:WA/AW-S-a}  \\
\multicolumn{1}{l|}{}                    & uR  &$n^2 \log n$   & \cite{mmkd-mwra-13}   & $n^2 \log n$   & Thm.\ref{thm:A-uR-a} & $n^2 \log n$      & Thm.\ref{thm:AW-uR-a} & $n^2 \log n$     &  Coro.\ref{coro:WA-uR-a} \\
\multicolumn{1}{l|}{}                    & R  &$n^2 \log n$   & \cite{mmkd-mwra-13}   & $n^3$          & Thm.\ref{thm:A-R-a} & $n^2 \log n + tn$ & Thm.\ref{thm:AW-R-a}  & $n^3$            & Coro.\ref{coro:WA-R-a}
\end{tabular}
\end{center}
\end{table}

Our results are summarized as follows.
\begin{itemize} \denseitems
 \item In a fixed orientation,
 a minimum-area square or uniform rectangular annulus enclosing
 a given set $P$ of $n$ points is shown to be
 one of also minimum width.
 On the other hand, we show that existing algorithms~\cite{ahimpr-bfr-03,ght-oafepramw-09}
 for the minimum-width variant
 indeed compute an annulus of minimum area.
 \item
 We show that, in a fixed orientation,
 the minimum-area rectangular annulus problem and
 the \emph{largest empty rectangle} problem are reduced to each other
 in linear time.
 The largest empty rectangle problem is known to have a lower bound of
 $\Omega(n\log n)$~\cite{mrs-flrop-85}.
 Hence, any algorithm computing a largest empty rectangle can be used to compute
 a minimum-area rectangular annulus in the same time bound.
% \item
% A uniform rectangular annulus and a square annulus of minimum area
% for a fixed orientation are shown to be those of minimum width, respectively,
% and vice versa.
% This implies that a minimum-area uniform rectangular annulus and
%square annulus for a fixed orientation can be computed by
% known algorithms for a minimum-width rectangular annulus and square annulus,
% which take time in $O(n)$ and in $O(n\log n)$ time,
% respectively~\cite{ahimpr-bfr-03,ght-oafepramw-09}.
 \item A minimum-area uniform rectangular annulus enclosing $P$
 in arbitrary orientation can be computed in $O(n^2 \log n)$ time
 by modifying the algorithm of Mukherjee et al.~\cite{mmkd-mwra-13}.
 \item A minimum-area rectangular annulus
 enclosing $P$ in arbitrary can be computed in $O(n^3)$ time
 by enumerating all \emph{maximal empty rectangles} among $P$ in all orientations
 using the algorithm by Chaudhuri et al.~\cite{cnd-lerps-03}.
 \item The above approach for the minimum-area rectangular annulus problem
 is shown to apply to other variants, resulting in the same time bound $O(n^3)$:
 computing a minimum-area square annulus,
 computing a minimum-width square annulus, and
 computing a largest empty square among $P$ in arbitrary orientation.
 Remarkably, for the minimum-width square annulus problem, we improve
 the previously best algorithm~\cite{b-cmwsaao-18} of running time $O(n^3 \log n)$
 by a factor of logarithm; for the largest empty square problem,
 no discussion can be found in the literature
 and therefore we present the first efficient algorithm.
 \item We present first algorithms for
  some bicriteria variants of the problem:
  computing a minimum-area minimum-width annulus or
  a minimum-width minimum-area annulus.
\end{itemize}
Table~\ref{tbl:results} summarizes
the running times of the currently best algorithms for the
\{minimum-width, minimum-area, minimum-area minimum-width,
minimum-width minimum-area\}
\{square, uniform rectangular, rectangular\} annulus problem in
\{fixed, arbitrary\} orientation.

%
%
%We start with considering the problem for the axis-parallel case, or equivalently
%the fixed-orientation case.
%For the square annulus, we show that a minimum-area square annulus
%is also a minimum-width annulus, and vice versa for a fixed orientation.
%For one definition, it is reduced to
%the \emph{largest empty rectangle} problem
%which has been extensively studied earlier~\cite{}.
%This observation leads to an algorithmic approach to the minimum-area annulus
%problem in arbitrary orientation,
%in which we handle the combinatorial description of all \emph{maximal empty rectangles} in all orientations.
%It is known by Das et al.~\cite{} there are at most $O(n^3)$
%combinatorially different maximum empty rectangles.
%Our algorithms for arbitrary orientation simply iterate each such maximum empty rectangle and minimize a certain function of a constant complexity over
%a domain of a constant complexity.
%This approach is shown to successfully solve the minimum-area rectangular and square annulus problem in arbitrary orientation,
%and even the minimum-width problems.

The rest of the paper is organized as follows:
after introducing necessary preliminaries in Section~\ref{sec:pre},
we consider the axis-parallel case of the problem in Section~\ref{sec:fixed_ori}.
We then address the problem of rectangular annuli
in arbitrary orientation in Section~\ref{sec:rect}, and
square annuli in Section~\ref{sec:square}.
We finally conclude the paper by Section~\ref{sec:conclusion}
with some remarks and open questions.

%The omitted proofs and additional figures will be provided in a full version.

%%%%%%%%%%%%%%%%%%%%%%%%%%%%%%%%%%%%%%%%%%%%%%%%%%%%%
\section{Preliminaries} \label{sec:pre}
%%%%%%%%%%%%%%%%%%%%%%%%%%%%%%%%%%%%%%%%%%%%%%%%%%%%

In this paper, we consider the plane $\Plane$ with
a standard coordinate system,
having the horizontal $x$-axis and the vertical $y$-axis.
For any subset $A \subseteq \Plane$, its boundary and interior,
denoted by $\bd A$ and $\intr A$, are defined
from the standard topology on $\Plane$.
For any two points $p, q\in \Plane$, let $\seg{pq}$ denote the line segment joining $p$ and $q$,
and $|\seg{pq}|$  denote the Euclidean length of $\seg{pq}$.

\paragraph{Orientations.}
The \emph{orientation} of a line or line segment $\ell$ in the plane is
a real number $\theta$ in range $[0, \pi)$
such that the rotated copy of the $x$-axis by $\theta$ counter-clockwise is parallel to $\ell$.
If the orientation of a line or line segment is $\theta$,
then we say that the line or line segment is \emph{$\theta$-aligned}.

A \emph{$\theta$-aligned rectangle (or square)} for $\theta \in [0, \pi/2)$ is
a rectangle (square, resp.)
whose sides are either $\theta$-aligned or $(\theta+\pi/2)$-aligned.
Note that any rectangle or square in the plane is $\theta$-aligned
for a unique $\theta \in [0, \pi/2)$.
Two rectangles or squares are said to be parallel if both are
$\theta$-aligned for some $\theta \in [0, \pi/2)$.
For a square, its \emph{center} is the intersection point of its two diagonals
and its \emph{radius} is half its side length.
Two squares are called \emph{concentric} if they are parallel and share
a common center.
We denote by $\area(R)$ the area of a rectangle or square $R$.

In each orientation $\theta \in [0, \pi/2)$, we regard any $\theta$-aligned line
to be \emph{horizontal} and directed in the $x$-coordinate increasing direction,
while any $(\theta+\pi/2)$-aligned line to be \emph{vertical} and directed
in the $y$-coordinate increasing direction.
For any $p,q\in \Plane$,
we say that $p$ is to the \emph{left} of $q$, or $q$ is to the \emph{right} of $p$, in orientation $\theta$
if the orthogonal projection of $p$ onto a $\theta$-aligned (directed) line
is prior to that of $q$.
Analogously, $p$ is \emph{below} $q$, or equivalently, $q$ is \emph{above} $p$
in $\theta$
if the orthogonal projection of $p$ onto a $(\theta+\pi/2)$-aligned (directed) line
is prior to that of $q$.
%For example, in \figurename~\ref{fig:width_annulus}(a),
%$p$ is to the left of and below $q$ in $\theta$.
This also enables us to identify the \emph{top}, \emph{bottom}, \emph{left},
and \emph{right} sides
of any rectangle or square in any orientation $\theta \in [0,\pi/2)$.

\paragraph{Square and rectangular annuli.}

A \emph{square annulus} $A$ is defined to be the closed region
between two concentric squares, called the \emph{outer} and \emph{inner squares}
of $A$.
Its \emph{width} $\width(A)$ is the difference of
the radii of its outer and inner squares,
and its \emph{area} $\area(A)$ is the difference of the areas
of its outer and inner squares.
See \figurename~\ref{fig:annulus}(a) for an illustration.

A \emph{rectangular annulus} $A$ is the closed region
between two parallel rectangles $R$ and $R'$
with $R' \subseteq R$, that is, $A = R \setminus \intr R'$.
The two rectangles $R$ and $R'$ defining $A$ as above
are called the \emph{outer} and \emph{inner rectangles} of $A$, respectively.
The perpendicular distance between the top side of $R$ and the top side of $R'$
is called the \emph{top-width} of $A$.
Analogously, we define the \emph{bottom-width}, \emph{left-width}, and
\emph{right-width} of rectangular annulus $A$.
The \emph{width} $\width(A)$ of $A$ is defined to be the \emph{maximum} of
its top-, bottom-, left-, and right-widths.
If the top-, bottom-, left-, and right-widths of $A$ are all equal,
then $A$ is called a \emph{uniform-width} (or, shortly \emph{uniform})
rectangular annulus.
The \emph{area} $\area(A)$ of $A$ is simply
the difference $\area(R) - \area(R')$ of the areas of
its outer and inner rectangles.
See \figurename~\ref{fig:annulus}(b--c) for an illustration.

A rectangular or square annulus is also called $\theta$-aligned
if its outer and inner rectangles or squares are $\theta$-aligned.
If an annulus is $0$-aligned, then it is rather called \emph{axis-parallel}.

\paragraph{Problem definition.}

In this paper, we study the problem of
computing a rectangular or square annulus of minimum extent that encloses
a given set $P$ of $n$ points in the plane
either in a fixed orientation or over all orientations $\theta \in [0, \pi/2)$.
The objective extents to minimize we are interested in are
\emph{width}, \emph{area}, or their combinations.

To be more precise, for each $\theta \in [0, \pi/2)$,
let $\mathcal{S}(\theta)$, $\mathcal{R}(\theta)$, and $\mathcal{R}_u(\theta)$
be the set of all $\theta$-alignes square, rectangular,
and uniform rectangular annuli, respectively, that enclose $P$.
By definition, note that $\mathcal{S}(\theta) \subset \mathcal{R}_u(\theta) \subset \mathcal{R}(\theta)$.
Also, let $\mathcal{S} := \bigcup_{\theta\in [0, \pi/2)} \mathcal{S}(\theta)$,
$\mathcal{R} := \bigcup_{\theta\in [0, \pi/2)} \mathcal{R}(\theta)$, and
$\mathcal{R}_u := \bigcup_{\theta\in [0, \pi/2)} \mathcal{R}_u(\theta)$.
For each class $\mathcal{A} \in \{\mathcal{S}, \mathcal{R}, \mathcal{R}_u\}$,
we consider the following optimization problems:
\begin{itemize} \denseitems
 \item The \emph{minimum-width} problem asks to minimize
 $\width(A)$ over all $A \in \mathcal{A}(\theta)$ for a fixed orientation
 $\theta \in [0, \pi/2)$, or
 over all $A \in \mathcal{A}$
 for arbitrary orientation.
 \item The \emph{minimum-area} problem asks to minimize
 $\area(A)$ over all $A \in \mathcal{A}(\theta)$ for a fixed orientation
 $\theta \in [0, \pi/2)$, or
 over all $A \in \mathcal{A}$
 for arbitrary orientation.
\end{itemize}
In most cases, there may be two or more annuli of a particular shape with
minimum width or minimum area.
This yields a series of \emph{bicriteria} optimization problems.
Here, in this paper, two variants are discussed.
\begin{itemize} \denseitems
 \item The \emph{minimum-area minimum-width} problem asks to find
 an optimal annulus whose area is the smallest among those with minimum width,
 either in a fixed orientation or over all orientations.
 \item The \emph{minimum-width minimum-area} problem asks to find
 an optimal annulus whose width is the smallest among those with minimum area,
 either in a fixed orientation or over all orientations.
\end{itemize}
%So, each variant of our problem is described by
%one of the four possible objectives,
%one of the three possible classes of shapes, and
%either in a fixed or arbitrary orientation, and will be
%shortly encoded by $\sigma_1$-$\sigma_2$-$\sigma_3$
%for $\sigma_1 \in \{\text{W}, \text{A}, \text{AW}, \text{WA} \}$,
%$\sigma_2 \in \{\text{S}, \text{R}, \text{uR}\}$, and
%$\sigma_3 \in \{\text{f}, \text{a}\}$,
%as in Table~\ref{tbl:results}.
%For instance, the A-R-a problem stands for the minimum-area rectangular annulus problem in arbitrary orientation.

%%%%%%%%%%%%%%%%%%%%%%%%%%%%%%%%%%%%%%%%%%%%%%%%%%%%%
\section{Annuli in Fixed Orientation} \label{sec:fixed_ori}
%%%%%%%%%%%%%%%%%%%%%%%%%%%%%%%%%%%%%%%%%%%%%%%%%%%%

In this section, we study the problem of computing a minimum-area
square and rectangular annulus in a given orientation $\theta\in [0,\pi/2)$,
and its variations.
Without loss of generality, we assume $\theta = 0$,
so our annuli in mind are all axis-parallel.
We first handle the square case and move on to the rectangular case.

\subsection{Axis-parallel square annulus}
Let $P$ be a set of $n$ input points in the plane~$\Plane$.
We show that any minimum-area axis-parallel square annulus
is indeed a \emph{minimum-width} axis-parallel square annulus.
For the purpose, we review a necessary observation about
minimum-width square annuli.
\begin{lemma}[Gluchshenko et al.~\cite{ght-oafepramw-09}]
 \label{lem:fixed_orient_minwidth_sa}
 There exists a minimum-width square annulus enclosing $P$
 such that its outer square is a smallest enclosing square for $P$,
 that is, a pair of its opposite sides contain a point of $P$ on each.
\end{lemma}

For the minimum-area case, we observe the following.
\begin{lemma}\label{lem:fixed_orient_minarea_sa}
 Let $A$ be any minimum-area axis-parallel square annulus enclosing $P$.
 Then, the outer square of $A$ is a smallest enclosing axis-parallel square
 for $P$.
\end{lemma}
\begin{proof}
Suppose for the contradiction that the outer square $S$ of $A$ is
not a smallest enclosing square for $P$.
Let $A'$ be a minimum-width square annulus such that
whose outer square $S'$ is a smallest enclosing square for $P$.
Such an annulus $A'$ exists by Lemma~\ref{lem:fixed_orient_minwidth_sa}.
We then have $\width(A) \geq \width(A')$ and
$\area(S) > \area(S')$.
This implies that $\area(A) > \area(A')$, a contradiction.
\end{proof}

Lemma~\ref{lem:fixed_orient_minarea_sa} indeed fixes the size of the outer square
of any possible minimum-area square annulus, and thus
the problem is reduced to finding such a minimum-width annulus described in
Lemma~\ref{lem:fixed_orient_minwidth_sa}.
\begin{lemma} \label{lem:A-S-f_W-S-f}
 Let $A$ be an axis-parallel square annulus enclosing $P$.
 Then, $A$ is of the minimum-area if and only if
 $A$ is of the minimum-width and its outer square is a smallest enclosing square
 for $P$.
\end{lemma}
The algorithm by Gluchshenko et al.~\cite{ght-oafepramw-09}
indeed computes a minimum-width square annulus whose outer square is
a smallest enclosing square for $P$ based on Lemma~\ref{lem:fixed_orient_minwidth_sa},
and they also proved a lower bound $\Omega(n \log n)$
for the minimum-width problem.
Hence, we conclude the following.
\begin{theorem} \label{thm:A-S-f}
 A minimum-area square annulus in a fixed orientation enclosing $n$ points
 can be computed in optimal $O(n\log n)$ time.
\end{theorem}
Note that the algorithm by Gluchshenko et al. can be easily modified to find
\emph{all} minimum-width square annuli whose outer square is
a smallest enclosing square for $P$ in $O(n \log n)$ time.
Thus,
a minimum-area minimum-width or minimum-width minimum-area square annulus
in a fixed orientation can be computed in the same time bound,
and this is time-optimal.
\begin{corollary} \label{coro:AW/WA-S-f}
 A minimum-area minimum-width and a minimum-width minimum-area square annulus
in a fixed orientation can be computed in optimal $O(n \log n)$ time.
\end{corollary}

\subsection{Axis-parallel rectangular annuli}
Next, we consider axis-parallel rectangular annuli.
The following is a well-known observation on the minimum-width
rectangular annuli.
\begin{lemma}[Abellanas et al.~\cite{ahimpr-bfr-03} and Mukherjee et al.~\cite{mmkd-mwra-13}]
 \label{lem:fixed_orient_minwidth_ra}
 There exists a minimum-width axis-parallel rectangular annulus enclosing $P$
 such that
 its outer rectangle is the smallest enclosing axis-parallel rectangle for $P$.
 The same holds for a minimum-width axis-parallel uniform rectangular annulus.
\end{lemma}
It is not difficult to see that for the minimum-width problem,
rectangular annuli and uniform rectangular annuli are equivalent.
Here, we give a direct and simple proof for the fact.
\begin{lemma} \label{lem:fixed_orient_minwidth_ura}
 Let $A$ be a minimum-width uniform rectangular annulus enclosing $P$.
 Then, $A$ is also a minimum-width rectangular annulus enclosing $P$.
\end{lemma}
\begin{proof}
Let $A'$ be a minimum-width rectangular annulus enclosing $P$.
Since any uniform rectangular annulus is also a rectangular annulus,
it holds that $\width(A') \leq \width(A)$ in general.
Suppose to the contrary that $\width(A') < \width(A)$,
and assume without loss of generality that the width $\width(A')$ of $A'$
is determined by the top-width of $A'$.
Then, by sliding each side of the inner rectangle of $A'$ inwards,
one can obtain a new uniform rectangular annulus $A''$ with
$\width(A'') = \width(A') < \width(A)$,
a contradiction.
\end{proof}

Unlike the minimum-width problem, the minimum-area problem behaves
differently for uniform and general rectangular annuli.
Our first observation on the minimum-area problem for rectangular annuli
is about their outer rectangles.
\begin{lemma}\label{lem:fixed_orient_minarea_ra}
 Let $A$ be a minimum-area axis-parallel rectangular annulus or
 a minimum-area uniform axis-parallel rectangular annulus enclosing $P$.
 In either case,
 its outer rectangle should be the smallest enclosing axis-parallel rectangle
 for $P$.
\end{lemma}
\begin{proof}
Suppose that the outer rectangle $R$ of $A$ is not
the smallest enclosing rectangle for $P$.
Then, replacing $R$ by the smallest enclosing rectangle for $P$
results in an annulus enclosing $P$ with smaller area,
a contradiction.
\end{proof}

This already implies the following.
\begin{theorem} \label{thm:A-uR-f}
 A minimum-area uniform rectangular annulus enclosing $P$ in a fixed orientation
 is unique and it coincides with the minimum-width
 uniform rectangular annulus whose outer rectangle is the smallest enclosing
 rectangle for $P$ in the orientation.
 Therefore, it can be computed in $O(n)$ time.
\end{theorem}
\begin{proof}
By Lemma~\ref{lem:fixed_orient_minarea_ra},
the outer rectangle of an axis-parallel minimum-area uniform rectangular annulus
enclosing $P$ should be the smallest axis-parallel rectangle $R$ enclosing $P$.
For a fixed outer rectangle $R$,
minimizing the area is equivalent to minimizing the width
over all uniform rectangular annuli enclosing $P$.
Hence, the axis-parallel minimum-area uniform rectangular annulus enclosing $P$
is uniquely determined by
the minimum-width axis-parallel uniform rectangular annulus whose outer rectangle
is $R$.
Such an annulus exists by Lemmas~\ref{lem:fixed_orient_minwidth_ra}
and~\ref{lem:fixed_orient_minwidth_ura},
and can easily computed in $O(n)$ time by computing the distance
from each point in $P$ to the boundary of the outer rectangle $R$,
as done in Abellanas~\cite{ahimpr-bfr-03}.
\end{proof}

By the uniqueness and the characterization of Theorem~\ref{thm:A-uR-f},
the minimum-area minimum-width and minimum-width minimum-area
uniform rectangular annulus in a fixed orientation can also be computed in
$O(n)$ time.
\begin{corollary} \label{coro:AW/WA-uR-f}
 A minimum-area minimum-width or minimum-width minimum-area
 uniform rectangular annulus enclosing $P$ in a fixed orientation
 can be computed in $O(n)$ time.
\end{corollary}

We then turn to the case of general rectangular annuli.
For the general rectangular annulus,
Lemma~\ref{lem:fixed_orient_minarea_ra} also fixes the outer rectangle
as the smallest enclosing rectangle $R$ for $P$.
Hence, the minimum-area problem is now reduced to
maximizing the area of the inner rectangle.
Let $R'$ be a possible inner rectangle that maximizes its area.
Then, $R'$ must satisfy the following conditions:
(1) $R'$ is \emph{empty} of points in $P$, that is, no point in $P$ lies in the interior of $R'$.
(2) $R'$ is contained in $R$, that is, $R' \subseteq R$.
Such a rectangle $R'$ satisfying conditions (1) and (2) is called
an \emph{empty rectangle} among $P$.
This directly shows a reduction to the problem of finding
an axis-parallel \emph{largest empty rectangle} among points $P$.
The largest empty rectangle problem takes an axis-parallel rectangle $B$
and a set $Q$ of $n$ points in $B$ as input, and
asks to find a maximum-area axis-parallel rectangle
that is empty of $P$ and is contained in $B$.

\begin{lemma} \label{lem:A-R-f-LER}
 The problem of computing a minimum-area axis-parallel rectangular annulus
 enclosing $P$ is computationally equivalent to the problem of computing
 a largest empty axis-parallel rectangle among $P$.
 More precisely, both problems are reduced to each other in linear time.
\end{lemma}
\begin{proof}
Let $P$ be an instance of the minimum-area axis-parallel rectangular annulus problem.
As discussed above, Lemma~\ref{lem:fixed_orient_minarea_ra} tells us
the outer rectangle of the solution should be the smallest enclosing rectangle $R$
for $P$.
Then, we can find the inner rectangle of maximum area
by solving the largest empty rectangle problem for input $R$ and $P$.
This shows a reduction to the largest empty rectangle problem in $O(n)$ time.

Next, consider an input $(B, Q)$ of the largest empty rectangle problem,
where $Q$ is a set of $n$ points in rectangle $B$.
To solve the problem, we apply any algorithm for the minimum-area rectangular
annulus problem for input $P := Q \cup \{p_1, p_2\}$,
where $p_1$ is the top-left corner and $p_2$ is the bottom-right corner of $B$,
and let $A$ be the output optimal annulus.
By Lemma~\ref{lem:fixed_orient_minarea_ra}, the outer rectangle of $A$
is $B$.
On the other hand, the inner rectangle $R'$ of $A$ should maximize its area
while $R'$ is empty of $P$ and $R' \subseteq B$.
Such a rectangle $R'$ is of course a largest empty rectangle.
This completes a reduction to the minimum-area rectangular annulus problem
in $O(n)$ time.
\end{proof}

Hence, our problem can be solved by applying any algorithm for
the largest empty rectangle problem.
The largest empty rectangle problem is one of the classical and well-studied problems in computational geometry.
There are two approaches to solve the problem:
whether one checks all \emph{maximal empty rectangles} or not.
Given a set $P$ of points and its bounding rectangle $R$,
a (axis-parallel) maximal empty rectangle is an empty rectangle
each of whose edges either contains a point in $P$ or a portion of an edge of $R$.
Naamad et al.~\cite{nlh-merp-84} presented an algorithm for the largest empty rectangle problem that runs in $O(\min\{r \log n, n^2\})$ time,
where $r$ denotes the number of maximal empty rectangles.
In the same paper, it is also shown that $r = O(n^2)$ in the worst case
and $r = O(n \log n)$ in expectation.
This algorithm was improved by Orlowski~\cite{o-nalerp-90}
to $O(n\log n + r)$ time.
Chazelle et al.~\cite{cdl-cler-86} presented
an $O(n \log^3 n)$-time algorithm that does not enumerate all maximal empty rectangles, and
Aggarwal and Suri~\cite{as-facler-87} improved it to $O(n\log^2 n)$ time
with $O(n)$ space.
Mckenna et al.~\cite{mrs-flrop-85} proved a lower bound of $\Omega(n \log n)$
for this problem.
\begin{theorem} \label{thm:A-R-f}
 A minimum-area rectangular annulus enclosing $P$ in a fixed orientation
 can be computed in $O(n \log^2 n)$ or $O(n\log n + r)$ time in the worst case, or
 in $O(n \log n)$ expected time,
 where $r = O(n^2)$ denotes the number of maximal empty rectangles among $P$.
 Moreover, any algorithm for the problem takes $\Omega(n \log n)$ time
 in the worst case.
\end{theorem}
\begin{proof}
The first statement directly follows from Lemma~\ref{lem:A-R-f-LER} and
the above discussion.
The currently fastest algorithms for the largest empty rectangle problem
are one by Aggarwal and Suri~\cite{as-facler-87} with running time $O(n\log^2 n)$
and the other by Orlowski~\cite{o-nalerp-90} with running time $O(n\log n + r)$.
As proved by Namaad et al.~\cite{nlh-merp-84},
$r = O(n^2)$ in the worst case and
$r = O(n\log n)$ in expectation,
so the algorithm by Orlowski runs in $O(n \log n)$ expected time.
Finally, again by Lemma~\ref{lem:A-R-f-LER}, the lower bound $\Omega(n \log n)$
by Mckenna et al.~\cite{mrs-flrop-85} applies to the minimum-area rectangular
annulus problem in a fixed orientation.
\end{proof}

A minimum-area minimum-width rectangular annulus can be found as follows:
First, compute the minimum-width uniform rectangular annulus $A_0$ in $O(n)$ time,
and let $w$ be its width and $R'_0$ be its inner rectangle.
We need to find a maximum-area empty rectangle $R'$ that contains $R'_0$
in order to keep the width of the annulus formed by $R$ and $R'$
at most $w$, that is, being the minimum width.
Since $A_0$ is a minimum-width annulus, we further observe that
at least one side of $R'$ must be overlapped with a side of $R'_0$.
Without loss of generality, we assume that the bottom side of $R'$ is a segment
on the line $\ell$ supporting the bottom side of $R'_0$.
Then, $R'$ can be found by enumerating all maximal empty rectangles
whose bottom side lies on $\ell$, testing if each contains $R'_0$,
and, if so, checking its area.
Fortunately, Orlowski~\cite[Section 2]{o-nalerp-90} showed how to
enumerate those maximal empty rectangles whose bottom side lies on $\ell$
in $O(n)$ time, after sorting the points in $P$ in their $x$- and $y$-coordinates.
Thus, we conclude the following.
\begin{theorem} \label{thm:AW-R-f}
 A minimum-area minimum-width rectangular annulus enclosing $P$
 in a fixed orientation
 can be computed in $O(n \log n)$ time.
\end{theorem}

A minimum-width minimum-area rectangular annulus can be found by
checking all largest empty rectangles for $P$.
However, known algorithms~\cite{as-facler-87,cdl-cler-86}
that do not enumerate all maximal empty rectangles
do not guarantee to find \emph{all} largest empty rectangles.
Hence, we apply the algorithm by Orlowski~\cite{o-nalerp-90}
and check every maximal empty rectangle and the width of the corresponding annulus.
This takes $O(n\log n + r)$ time.
\begin{theorem} \label{thm:WA-R-f}
 A minimum-width minimum-area rectangular annulus enclosing $P$
 in a fixed orientation
 can be computed in $O(n \log n + r)$ time,
 where $r$ denotes the number of maximal empty rectangles among $P$.
\end{theorem}

%%%%%%%%%%%%%%%%%%%%%%%%%%%%%%%%%%%%%%%%%%%%%%%%%%%%%
\section{Rectangular Annuli in Arbitrary Orientation} \label{sec:rect}
%%%%%%%%%%%%%%%%%%%%%%%%%%%%%%%%%%%%%%%%%%%%%%%%%%%%

In this section, we consider the problem of computing
an optimal rectangular annulus in arbitrary orientation for different objectives.

\begin{figure}[tb]
\begin{center}
\includegraphics[width=.8\textwidth]{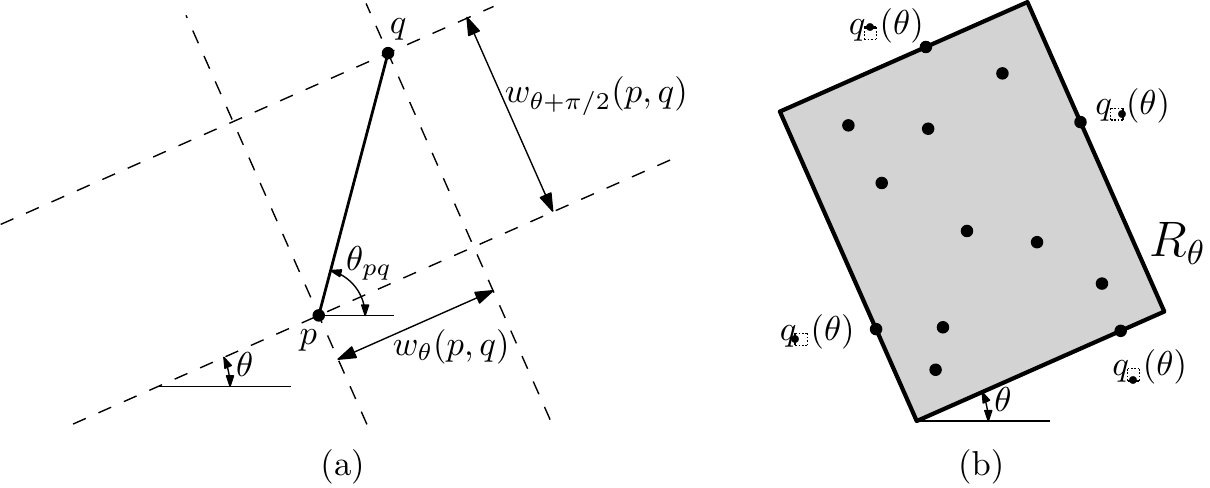}
\end{center}
%\vspace{-5mm}
\caption{For $\theta \in [0, \pi/2)$,
 (a) $w_\theta(p, q)$ and $w_{\theta+\pi/2}(p,q)$
 for two points $p, q\in \Plane$, and
 (b) the smallest enclosing $\theta$-aligned rectangle $R_\theta$ for points $P$.
 }
% \vspace{-5mm}
\label{fig:w_R}
\end{figure}

We will often discuss the distance between
the orthogonal projections of $p$ and $q$ onto a $\theta$-aligned line
for $\theta \in [0, \pi)$,
denoted by $w_\theta(p, q)$.
It is not difficult to see that
 \[ w_\theta(p, q) = |\seg{pq}| \cdot | \cos (\theta_{pq} - \theta)|,\]
where $\theta_{pq} \in [0,\pi)$ denotes the orientation of $\seg{pq}$.
Intuitively, one may regard that
$w_\theta(p,q)$ is the horizontal distance
and $w_{\theta+\pi/2}(p,q)$ is the vertical distance between $p$ and $q$
in orientation $\theta \in [0,\pi/2)$.
See \figurename~\ref{fig:w_R}(a).

%%%

As the side length of a $\theta$-aligned rectangle or square of our interests
is often described by a combination of $w_\theta$ with different
pairs of points $(p, q)$,
we first introduce some basic properties of such functions of variable $\theta$.
A function of a particular form
$a \sin(\omega \theta + \phi) + b$
for some constants $a, b, \omega, \phi \in \Real$
is called a \emph{sinusoidal function},
where $a$, $b$, $\omega$ and $\phi$ are called its \emph{amplitude},
\emph{base}\footnote{%
The constant $b$ is usually called the \emph{center amplitude}
in wave analysis, while here we call it base for simplicity.},
\emph{(angular) frequency} and
\emph{phase}, respectively.
Note that $\cos(\theta)$ is also sinusoidal with base $0$
as $\cos(\theta) = \sin(\theta + \pi/2)$.
Observe that for a fixed pair of points $p, q\in \Plane$,
$w_\theta(p, q)$
is a piecewise sinusoidal function of base $0$ and frequency $1$
over $\theta \in [0, \pi)$
with at most one breakpoint at $\theta = \theta_{pq} - \pi/2$ or $\pi/2-\theta_{pq}$.

The following properties of sinusoidal functions are well known
and easily derived.
\begin{observation} \label{obs:sinusoidal_sum}
 The sum of two sinusoidal functions of base $0$ and frequency $\omega$
 is also sinusoidal of base $0$ and frequency $\omega$.
 Therefore, the graphs of two sinusoidal functions of equal frequency
 cross at most once in domain $[0, \pi)$.
 %\hfill \copy\ProofSym
\end{observation}
\begin{observation} \label{obs:sinusoidal_product}
 The product of two sinusoidal functions of base $0$ and frequency $\omega$
 is equal to
 a sinusoidal functions of frequency $2\omega$.
 Specifically, it holds that
 \[ a_1\sin(\omega \theta + \phi_1)\cdot a_2\sin(\omega\theta + \phi_2) =
 \frac{a_1a_2}{2}\left(\sin\left(2\omega\theta + \phi_1 +\phi_2 - \frac{\pi}{2}\right) + \cos(\phi_1-\phi_2)\right),\]
 for any constants $a_1, a_2, \phi_1, \phi_2, \omega \in \Real$.
\end{observation}

\subsection{Uniform rectangular annuli}
%Let $P$ be a given set of $n$ points in $\Plane$.
For any orientation $\theta \in [0, \pi/2)$,
let $R_\theta$ be the smallest $\theta$-aligned rectangle enclosing $P$.
The rectangle $R_\theta$ is determined by the
leftmost, rightmost, topmost, and bottommost points in orientation $\theta$
among those in $P$, in such a way that each side of $R_\theta$ contains
each of these four points.
Let $\extr_\mt(\theta), \extr_\mb(\theta), \extr_\ml(\theta), \extr_\mr(\theta) \in P$
be the topmost, bottommost, leftmost, and rightmost points in orientation $\theta$
among those in $P$.
(If there are two or more topmost, bottommost, leftmost, or rightmost points
in $\theta$, then we choose an arbitrary one.)
See \figurename~\ref{fig:w_R}(b).
For each $p\in P$ and $\theta \in [0, \pi/2)$, we define a function $f_p \colon [0,\pi/2) \to \Real$ to be the ``distance''
from the boundary of $R_\theta$.
More precisely, $f_p(\theta)$ takes the minimum among the perpendicular distances
from every side of $R_\theta$ to $p$:
 \[ f_p(\theta) := \min\{w_{\theta+\pi/2}(p, \extr_\mt(\theta)), w_{\theta+\pi/2}(p, \extr_\mb(\theta)),
  w_{\theta}(p, \extr_\ml(\theta)), w_{\theta}(p, \extr_\mr(\theta))\}.\]

Now, we let $w(\theta)$ be the width of a minimum-width $\theta$-aligned
rectangular annulus enclosing $P$.
By Lemma~\ref{lem:fixed_orient_minwidth_ra}, there exists
a minimum-width uniform $\theta$-aligned rectangular annulus
whose outer rectangle is $R_\theta$.
Note that such a uniform rectangular annulus is unique
and we denote it by $A_u(\theta)$.
Moreover, Lemma~\ref{lem:fixed_orient_minwidth_ura} states that
$A_u(\theta)$ is also a minimum-width $\theta$-aligned rectangular annulus.
Hence, we have
\[ w(\theta) = \width(A_u(\theta)) = \max_{p\in P} f_p(\theta).\]

The above discussion already results in an $O(n^2 \log n)$-time
algorithm for the minimum-width rectangular annulus problem
over all orientations $\theta \in [0,\pi/2)$.
The minimum-width problem can be solved by
computing $\min_{\theta\in [0,\pi/2)} w(\theta)$ and the corresponding annulus.
Observe that for an interval $I \subset [0, \pi/2)$ such that
the tuple of four points
 $(\extr_\mt(\theta), \extr_\mb(\theta), \extr_\ml(\theta), \extr_\mr(\theta))$
is fixed,
each function $f_p$ restricted to $I$ is piecewise sinusoidal
with at most eight pieces.
In such an interval $I$, minimizing $w$ over $\theta \in I$
can be done by computing the upper envelope of functions $f_p$
and finding a lowest point in the envelope.
\begin{lemma} \label{lem:width_ra}
 Let $I \subset [0,\pi/2)$ be an interval in which the four points
 $\extr_\mt(\theta), \extr_\mb(\theta), \extr_\ml(\theta), \extr_\mr(\theta)$
 determining $R_\theta$ are fixed.
 Then, the function $w$ over $I$ is a piecewise frequency-$1$ sinusoidal function
 of complexity $O(n \alpha(n))$ and can be explicitly computed in $O(n \log n)$ time.
\end{lemma}
\begin{proof}
Note that function $w$ is the upper envelope of $n$ functions $f_p$ for $p\in P$.
Since $f_p$ over $I$ consists of at most eight sinusoidal curves
for each $p\in P$,
function $w$ over $I$ is the upper envelope of $O(n)$ curves
that are sinusoidal of frequency $1$.
By Observation~\ref{obs:sinusoidal_sum}, any two of these curves
cross at most once.
Thus, the upper envelope $w$ on $I$ can be explicitly computed in
$O(n \log n)$ time~\cite{h-fuenls-89},
while $w$ consists of $O(n \alpha(n))$ sinusoidal pieces~\cite{sa-dsstga-95}.
\end{proof}
Finally, as observed earlier by Toussaint~\cite{t-sgprc-83}
and also applied by Mukherjee et al.~\cite{mmkd-mwra-13} and Bae~\cite{b-cmwsaao-18},
the domain $[0, \pi/2)$ is decomposed into at most $O(n)$ such intervals $I$.
\begin{lemma}[Toussaint~\cite{t-sgprc-83}]
 \label{lem:primary_interval}
 There are at most $O(n)$ changes in the tuple of four points
 $(\extr_\mt(\theta),$ $\extr_\mb(\theta),$ $\extr_\ml(\theta),$ $\extr_\mr(\theta))$
 as a function of $\theta \in [0,\pi/2)$.
 That is, the domain $[0, \pi/2)$ is partitioned into $O(n)$ maximal intervals,
 in each of which the four points stay constant.
\end{lemma}
Each such interval described in Lemma~\ref{lem:primary_interval}
is called \emph{primary}.
Therefore, in $O(n^2 \log n)$ time, we can compute
a minimum-width rectangular annulus enclosing $P$ over all orientations\footnote{%
The $O(n^2 \log n)$-time algorithm presented
by Mukherjee et al.~\cite{mmkd-mwra-13}
is slightly different, but its outline and approach is almost identical to
that we describe here.
}.

Note that the above algorithm
indeed finds \emph{all} minimum-width uniform rectangular annulus enclosing $n$
points over all orientations.
Hence, by picking one with minimum area among all of them,
we can find a minimum-area minimum-width uniform rectangular annulus.
\begin{theorem} \label{thm:AW-uR-a}
 A minimum-area minimum-width uniform rectangular annulus enclosing $n$ points
 over all orientations can be computed in $O(n^2 \log n)$ time.
\end{theorem}

Now, we turn to the minimum-area problem for uniform rectangular annuli.
Recall Theorem~\ref{thm:A-uR-f} stating that $A_u(\theta)$
is indeed the unique minimum-area $\theta$-aligned uniform rectangular annulus
enclosing $P$.
So, we want to minimize $\area(A_u(\theta))$ over $\theta \in [0,\pi/2)$.
The area $\area(A_u(\theta))$ can be represented as follows:
\begin{align*}
 \area(A_u(\theta)) & = \peri(R_\theta) \cdot w(\theta) - 4(w(\theta))^2 \\
                 & = 2(w_{\theta+\pi/2}(\extr_\mt(\theta), \extr_\mb(\theta))+w_{\theta}(\extr_\ml(\theta), \extr_\mr(\theta)))\cdot w(\theta) - 4(w(\theta))^2,
\end{align*}
where $\peri(\cdot)$ denotes the perimeter of a rectangle.

Consider any primary interval $I\subset [0,\pi/2)$.
In $I$, the four extreme points $\extr_\mt = \extr_\mt(\theta)$, $\extr_\mb = \extr_\mb(\theta)$, $\extr_\ml = \extr_\ml(\theta)$, and $\extr_\mr = \extr_\mr(\theta)$ are fixed,
so each of the terms $w_{\theta+\pi/2}(\extr_\mt, \extr_\mb)$ and
$w_{\theta}(\extr_\ml, \extr_\mr)$ are piecewise base-$0$ frequency-$1$
sinusoidal, and so is their sum by Observation~\ref{obs:sinusoidal_sum}.
Further, Lemma~\ref{lem:width_ra} tells us that
$w(\theta)$ is also piecewise base-$0$ frequency-$1$ sinusoidal over $I$.
From the above discussion, we then observe that $\area(A_u(\theta))$
is piecesise frequency-$2$ sinusoidal by Observation~\ref{obs:sinusoidal_product}.
Hence, we can compute an explicit description of the function $\area(A_u(\theta))$
over $\theta \in I$, provided that the description of function $w$
has been computed,
in time proportional to the number of breakpoints of $w$.
\begin{theorem} \label{thm:A-uR-a}
 A minimum-area uniform rectangular annulus enclosing $n$ points
 over all orientations can be computed in $O(n^2 \log n)$ time.
\end{theorem}
\begin{proof}
We handle each primary interval $I\subset [0,\pi/2)$ in which
the four extreme points $\extr_\mt = \extr_\mt(\theta)$, $\extr_\mb = \extr_\mb(\theta)$, $\extr_\ml = \extr_\ml(\theta)$, and $\extr_\mr = \extr_\mr(\theta)$ are fixed.
Apply Lemma~\ref{lem:width_ra} to explicitly compute function $w$ on $I$
in $O(n\log n)$ time.
Note that $w$ is a piecewise sinusoidal function of $O(n \alpha(n))$ breakpoints.
As briefly discussed above, the term
$w_{\theta+\pi/2}(\extr_\mt, \extr_\mb)+w_{\theta}(\extr_\ml, \extr_\mr)
= \peri(R_\theta)/2$
is a piecewise sinusoidal function of $O(1)$ breakpoints.
Compute a sorted list of all breakpoints of $w$ and $\peri(R_\theta)$
in $O(n \alpha(n))$ time.
This is possible because the breakpoints of $w$ are already given sorted
from Lemma~\ref{lem:width_ra}.
We then compute the function $\area(A_u(\theta))$ between
every two consecutive breakpoints and find a minimum in there.
Note that $\area(A_u(\theta))$ between two consecutive breakpoints
is a sinusoidal function of frequency $2$ by Observation~\ref{obs:sinusoidal_product} and can be minimized in $O(1)$ time.

Consequently, we can minimize $\area(A_u(\theta))$ over $\theta \in I$
in $O(n \log n)$ time.
Since there are $O(n)$ primary intervals by Lemma~\ref{lem:primary_interval},
our algorithm takes $O(n^2 \log n)$ time in total.
\end{proof}

Remark that our algorithm described in Theorem~\ref{thm:A-uR-a}
indeed finds \emph{all} minimum-area uniform rectangular annulus enclosing $n$
points over all orientations.
Hence, by selecting one with minimum width among all of them,
we can find a minimum-width minimum-area uniform rectangular annulus.
\begin{corollary} \label{coro:WA-uR-a}
 A minimum-width minimum-area uniform rectangular annulus enclosing $n$ points
 over all orientations can be computed in $O(n^2 \log n)$ time.
\end{corollary}

\subsection{Minimum-area minimum-width rectangular annulus}
Next, we solve the minimum-area minimum-width rectangular annulus problem.
For the purpose, we specify all orientations in which the width function
$w(\theta)$ is minimized over $\theta \in [0, \pi/2)$,
and then we find a minimum-area rectangular annulus
in each of these orientations, applying the algorithm of Theorem~\ref{thm:AW-R-f}.

\begin{theorem} \label{thm:AW-R-a}
 A minimum-area minimum-width rectangular annulus enclosing $n$ points
 over all orientations can be computed in $O(n^2\log n + tn)$ time,
 where $t=O(n^2)$ denotes the number of different orientations $\theta'$
 such that $w(\theta') = \min_{\theta \in [0,\pi/2)} w(\theta)$.
\end{theorem}
\begin{proof}
As described above, we first compute the explicit description of function $w$
over $\theta \in [0, \pi/2)$ in $O(n^2 \log n)$ time
as described in the previous section by Lemmas~\ref{lem:width_ra}
and~\ref{lem:primary_interval}.
Since we identify the full description of $w$,
we can specify all orientations that minimize $w$.
Let $\Theta$ be the set of all orientations $\theta \in [0, \pi/2)$ that
minimize $w$ in the same time bound.
Note that the orientations in $\Theta$ are obtained in the sorted order.
We then apply the algorithm of Theorem~\ref{thm:AW-R-f} for each $\theta \in \Theta$.
Provided the sorted lists of $P$ in orientation $\theta$ and $\theta + \pi/2$,
this takes $O(n)$ time for each $\theta \in \Theta$.
As $\theta$ continuously increases from $0$ to $\pi/2$,
the sorted lists of $P$ in orientation $\theta$ and $\theta + \pi/2$
can be maintained
in total $O(n^2 \log n)$ time
after sorting the orientations of all lines through two points in $P$
in $O(n^2 \log n)$ time.
Thus, a minimum-area minimum-width rectangular annulus enclosing $P$
can be found in $O(n^2 \log n + tn)$ time,
where $t = |\Theta|$.

Now, we bound the cardinality $t$ of $\Theta$.
As proved in Lemma~\ref{lem:width_ra},
the complexity of $w$ in a primary interval $I$ is $O(n \alpha(n))$.
This already implies that $t = O(n^2 \alpha(n))$ by Lemma~\ref{lem:primary_interval}.
Here, we prove that $t = O(n^2)$ in fact.

Let $w^* := \min_{\theta \in [0,\pi/2)} w(\theta)$ be the minimum width
over $\theta \in [0, \pi/2)$.
In each primary interval $I \subset [0, \pi/2)$, recall that
the width function $w$ is the upper envelope of $n$ functions $f_p$ and
$f_p$ is a piecewise frequency-$1$ sinusoidal function of $O(1)$ breakpoints.
By Observation~\ref{obs:sinusoidal_sum},
for each $p\in P$, there are $O(1)$ different orientations $\theta \in I$
such that $f_p(\theta) = w^*$.
Hence, we have $|\Theta \cap I| = O(n)$
for any primary interval $I \subset [0, \pi/2)$,
and $t = |\Theta| = O(n^2)$ by Lemma~\ref{lem:primary_interval}.
\end{proof}

\subsection{Minimum-area rectangular annulus}
Finally, we consider the minimum-area problem for general rectangular annuli
in arbitrary orientation.
For a fixed orientation $\theta \in [0, \pi/2)$,
recall that the problem is equivalent to the largest empty rectangle problem
in orientation $\theta$ as shown in Lemma~\ref{lem:A-R-f-LER},
so we can find a largest empty $\theta$-aligned rectangle in
the smallest enclosing $\theta$-aligned rectangle $R_\theta$ for $P$.

Let $a(\theta)$ be the area of a minimum-area $\theta$-aligned rectangular
annulus enclosing $P$.
Our problem here is to minimize $a(\theta)$ over $\theta \in [0,\pi/2)$
and to find the corresponding annulus.
By Lemma~\ref{lem:A-R-f-LER}, the inner rectangle of an optimal annulus
over all orientations is a maximal empty rectangle among $P$ in some orientation.
A maximal empty rectangle among $P$ \emph{in arbitrary orientation}
is defined to be a $\theta$-aligned empty rectangle $E_\theta$
for $\theta \in [0,\pi/2)$ such that
$E_\theta$ is contained in $R_\theta$ and
there is no other $\theta$-aligned empty rectangle $E'_\theta$
with $E_\theta \subset E'_\theta$.
This implies that each side of any maximal empty rectangle
either contains a point of $P$ or is a portion of a side of $R_\theta$.
We then represent a maximal empty rectangle $E_\theta$ by
a tuple $(p_\mt, p_\mb, p_\ml, p_\mr; \theta)$,
where $\theta$ is the orientation of $E_\theta$, and
\begin{itemize} \denseitems
 \item $p_\mt$ denotes a point in $P$ lying on the top side of $E$, if any,
  or $p_\mt = \extr_\mt(\theta)$, otherwise;
 \item $p_\mb$ denotes a point in $P$ lying on the bottom side of $E$, if any,
  or $p_\mb = \extr_\mb(\theta)$, otherwise;
 \item $p_\ml$ denotes a point in $P$ lying on the left side of $E$, if any,
  or $p_\ml = \extr_\ml(\theta)$, otherwise;
 \item $p_\mr$ denotes a point in $P$ lying on the right side of $E$, if any,
  or $p_\mr = \extr_\mr(\theta)$, otherwise.
\end{itemize}

Two maximal empty rectangles $E_\theta = (p_\mt, p_\mb, p_\ml, p_\mr; \theta)$
and $E'_{\theta'} = (p'_\mt, p'_\mb, p'_\ml, p'_\mr; \theta')$
with $\theta\leq \theta'$
are said to be \emph{combinatorially equivalent}
if $p_\ma = p'_\ma$ for each $\ma \in \{\mt, \mb,\ml, \mr\}$
and all rectangles represented by $(p_\mt, p_\mb, p_\ml, p_\mr; \theta'')$
for $\theta \leq \theta'' \leq \theta'$ are maximal empty rectangles among $P$.
We call a maximal set of combinatorially equivalent maximal empty rectangles
among $P$
a \emph{MER class}.
A MER class is thus
represented by a tuple $(p_\mt, p_\mb, p_\ml, p_\mr; J)$,
where $J \subseteq [0, \pi/2)$ is an orientation interval such that
$(p_\mt, p_\mb, p_\ml, p_\mr; \theta)$ is a maximal empty rectangle
for each $\theta \in J$.
We shall call $J$ the \emph{valid interval} of MER class $E$.
Chaudhuri et al.~\cite{cnd-lerps-03} proved that there are $O(n^3)$
different MER classes,
and presented an $O(n^3)$-time algorithm to specify all of them.

Let $E = (p_\mt, p_\mb, p_\ml, p_\mr; J)$ be a MER class,
and $E_\theta := (p_\mt, p_\mb, p_\ml, p_\mr; \theta)$ for $\theta \in J$
be the maximal empty $\theta$-aligned rectangle in $E$.
Then, the area $\area(E_\theta)$ of $E_\theta$ is
 \[ \area(E_\theta) = w_{\theta+\pi/2}(p_\mt, p_\mb)\cdot w_\theta(p_\ml, p_\mr),\]
which is the product of two sinusoidal functions of base $0$ and frequency $1$.
Define $a_E(\theta)$ be the area of the rectangular annulus
whose outer rectangle is $R_\theta$ and inner rectangle is $E_\theta$.
We then have
\begin{align*}
 a_E(\theta) & = \area(R_\theta) - \area(E_\theta) \\
           &=  w_{\theta+\pi/2}(\extr_\mt(\theta), \extr_\mb(\theta))\cdot
             w_\theta(\extr_\ml(\theta), \extr_\mr(\theta))
            - w_{\theta+\pi/2}(p_\mt, p_\mb)\cdot w_\theta(p_\ml, p_\mr).
\end{align*}
Since the four extreme points
$\extr_\mt(\theta), \extr_\mb(\theta), \extr_\ml(\theta), \extr_\mr(\theta)$
are fixed in a primary interval $I \subset [0,\pi/2)$,
the function $a_E(\theta)$ over $\theta \in I\cap J$ is indeed
a piecewise sinusoidal function of frequency $2$
with $O(1)$ breakpoints
by Observations~\ref{obs:sinusoidal_sum} and~\ref{obs:sinusoidal_product},
so can be minimized in $O(1)$ time.

From the observations above, we are ready to describe
our algorithm for the minimum-area rectangular annulus problem.
First, specify all primary intervals in $O(n \log n)$ time by
Toussaint~\cite{t-sgprc-83} and all MER classes in $O(n^3)$ time
using the algorithm by Chaudhuri et al.~\cite{cnd-lerps-03}.
For each MER class $E$ with valid interval $J$,
we call a primary interval $I$ \emph{relevant to $E$}
if $I \cap J \neq \emptyset$.
We then compute all relevant primary intervals for each MER class
in $O(n^3)$ time as follows.
\begin{lemma} \label{lem:relevant_primary_int}
 The number of pairs $(E, I)$ such that $E$ is a MER class and
 $I$ is a primary interval relevant to $E$
 is $O(n^3)$,
 and all such pairs can be computed in $O(n^3)$ time.
\end{lemma}
\begin{proof}
Let $T$ be the number of such pairs $(E, I)$.
Since the algorithm by Chaudhuri et al.~\cite{cnd-lerps-03}
runs in an orientation-sweeping fashion,
we can easily specify and store all primary intervals $I$
relevant to each MER class $E$ in time $O(n^3 + T)$.

Now, we show that $T = O(n^3)$, so the running time above
is bounded also by $O(n^3)$.
Consider any two consecutive primary intervals $I$ and $I'$,
and let $\theta$ be the shared endpoint of $I$ and $I'$.
Naamad et al.~\cite{nlh-merp-84} showed that the number of
maximal empty rectangles in a fixed orientation $\theta$ is $O(n^2)$.
This implies that there are at most $O(n^2)$ MER classes $E$ such that
both $I$ and $I'$ are relevant to $E$.
Since there are $O(n)$ primary intervals by Lemma~\ref{lem:primary_interval}
and $O(n^3)$ MER classes shown by Chaudhuri et al.~\cite{cnd-lerps-03},
we conclude that $T = O(n^2\cdot n + n^3) = O(n^3)$.
\end{proof}

Then, for each MER class $E = (p_\mt, p_\mb, p_\ml, p_\mr; J)$ and
each primary interval $I$ relevant to $E$,
we construct the function $a_E(\theta)$ and minimize it
over $\theta \in I \cap J$ in $O(1)$ time.
The minimum of such minima over all MER classes and relevant primary intervals
is indeed the minimum possible area of a rectangular annulus enclosing $P$
since $a(\theta) = \min_{E} a_E(\theta)$ is the lower envelope
of $a_E$ over all MER classes $E$ and we want to compute
 \[ \min_{\theta \in [0,\pi/2)} a(\theta)
   = \min_{\theta \in [0,\pi/2)} \min_E a_E(\theta)
   = \min_E \min_{\theta \in [0,\pi/2)} a_E(\theta).
\]
This proves the correctness of our algorithm.
Hence, we conclude the following theorem.

\begin{theorem} \label{thm:A-R-a}
 A minimum-area rectangular annulus enclosing $n$ points
 over all orientations can be computed in $O(n^3)$ time.
\end{theorem}

Our algorithm described above indeed can find \emph{all}
minimum-area rectangular annuli enclosing $P$ over all orientations.
Therefore, a minimum-width minimum-area rectangular annulus can be found
by selecting one with minimum width over all minimum-area annuli.
\begin{corollary} \label{coro:WA-R-a}
 A minimum-width minimum-area rectangular annulus enclosing $n$ points
 over all orientations can be computed in $O(n^3)$ time.
\end{corollary}

%%%%%%%%%%%%%%%%%%%%%%%%%%%%%%%%%%%%%%%%%%%%%%%%%%%%%
\section{Square Annuli in Arbitrary Orientation} \label{sec:square}
%%%%%%%%%%%%%%%%%%%%%%%%%%%%%%%%%%%%%%%%%%%%%%%%%%%%
In this section, we consider the square annulus problem
in arbitrary orientation.

Let $a(\theta)$ be the area of a minimum-area $\theta$-aligned square annulus
enclosing $P$,
and $w(\theta)$ be the width of a minimum-width $\theta$-aligned square annulus
enclosing $P$.
By Lemma~\ref{lem:fixed_orient_minarea_sa}, the outer square
of a minimum-area $\theta$-aligned square annulus should be
a smallest enclosing $\theta$-aligned square for $P$.
Let $d(\theta)$ be the side length of a smallest enclosing $\theta$-aligned square for $P$.
We have $d(\theta) = \max\{w_{\theta+\pi/2}(\extr_\mt(\theta), \extr_\mb(\theta)),
 w_\theta(\extr_\ml(\theta), \extr_\mr(\theta))\}$
since a smallest enclosing $\theta$-aligned square for $P$
contains the smallest enclosing $\theta$-aligned rectangle $R_\theta$ for $P$.
We then have
 \[ a(\theta) = 4d(\theta)w(\theta) - 4(w(\theta))^2,\]
by Lemma~\ref{lem:A-S-f_W-S-f}.

The author in~\cite{b-cmwsaao-18} proved that
the width function $w(\theta)$ on $[0,\pi/2)$ is piecewise
base-$0$ frequency-$1$ sinusoidal with $O(n^3)$ breakpoints,
and showed how to compute its explicit description in $O(n^3 \log n)$ time.
Using this, we have the following result.
\begin{lemma} \label{lem:sa_area}
 The function $a \colon [0,\pi/2) \to \Real$ is piecewise
 frequency-$2$ sinusoidal with $O(n^3)$ breakpoints,
 and can be explicitly computed in $O(n^3 \log n)$ time.
\end{lemma}
\begin{proof}
As introduced above, by the results in~\cite{b-cmwsaao-18},
$w \colon [0,\pi/2) \to \Real$ is piecewise base-$0$ frequency-$1$ sinusoidal
with $O(n^3)$ breakpoints and can be explicitly computed in $O(n^3 \log n)$ time.
On the other hand, the function $d \colon [0,\pi/2) \to \Real$
is piecewise base-$0$ frequency-$1$ sinusoidal with $O(n)$ breakpoints
by Lemma~\ref{lem:primary_interval},
since $d$ has at most $O(1)$ breakpoints in each primary interval
by Observation~\ref{obs:sinusoidal_sum}.

Now, recall that $a(\theta) = 4d(\theta)w(\theta) - 4(w(\theta))^2$.
So, each breakpoint of $a$ is either a breakpoint of $w$ or one of $d$.
Between two consecutive breakpoints of $a$,
both $d$ and $w$ are sinusoidal functions of frequency $1$ and base $0$.
As $a$ is the sum of products of two such sinusoidal functions,
it appears to be a sinusoidal function of frequency $2$
by Observations~\ref{obs:sinusoidal_sum} and~\ref{obs:sinusoidal_product}.

After gathering the breakpoints of $w$ and $d$ and sorting them,
we can obtain an explicit description of $a$ in $O(1)$ time per breakpoint.
Hence, $O(n^3 \log n)$ time is sufficient.
\end{proof}

This already implies that, in $O(n^3 \log n)$ time,
we can indeed minimize the area function $a(\theta)$
over $\theta \in [0,\pi/2)$ and obtain
a minimum-area square annulus enclosing $P$ in arbitrary orientation.
In the following, we show how to improve this downto $O(n^3)$
without computing a full description of the area function $a$.

For the purpose, we make use of the maximal empty rectangles
in arbitrary orientation,
as done above for the minimum-area rectangular annulus problem.
Our approach here is based on an easy observation that
the inner square of a minimum-area $\theta$-aligned square annulus
should be contained in a maximal empty $\theta$-aligned rectangle.
We thus check each MER class $E = (p_\mt, p_\mb, p_\ml, p_\mr; J)$,
searching a minimum-area square annulus whose inner square is contained
in the maximal empty rectangle $E_\theta = (p_\mt, p_\mb, p_\ml, p_\mr; \theta)$
for some $\theta \in J$.
More precisely, we define $a_E(\theta)$ for $\theta \in J$ to be the
minimum possible area of square annuli whose outer square is
a smallest enclosing $\theta$-aligned square for $P$
and whose inner square is contained in $E_\theta$.
By definition, possible locations of inner squares are restricted;
if there is no possible center for an inner square, then
$a_E(\theta)$ is defined to be the area of the outer square,
that is, $a_E(\theta) = (d(\theta))^2$.
Details are given below.

%Let $\delta(\theta)$ be the side length of the inner square
%of a minimum-area $\theta$-aligned square annulus enclosing $P$.
%By Lemma~\ref{lem:fixed_orient_minarea_sa},
%we have $a(\theta) = (d(\theta))^2 - (\delta(\theta))^2$.
%Let $C_\theta$ be the set of centers of all smallest $\theta$-aligned square
%enclosing $P$.
%Observe that $C_\theta$ forms a line segment whose orientation is
%$\theta$ if $w_{\theta+\pi/2}(\extr_\mt(\theta), \extr_\mb(\theta))
% > w_\theta(\extr_\ml(\theta), \extr_\mr(\theta))$ or
%$\theta+\pi/2$, otherwise.
%Further, its length $|C_\theta|$ is equal to
%$|w_{\theta+\pi/2}(\extr_\mt(\theta), \extr_\mb(\theta))
% - w_\theta(\extr_\ml(\theta), \extr_\mr(\theta))|$.
%Again Lemma~\ref{lem:fixed_orient_minarea_sa} implies that
%any minimum-area $\theta$-aligned square annulus is centered at
%some $c \in C_\theta$, and so is its inner square.
%For any $c\in \Plane$, we define
%\[ \delta(\theta, c) :=
%    \begin{cases}
%      \min_{p\in P} \max\{w_\theta(p, c), w_{\theta+\pi/2}(p, c)\}
%          & \text{if $c \in C_\theta$}\\
%      0 & \text{otherwise}
%    \end{cases}
%\]
%to be the side length of the largest empty $\theta$-aligned square centered at $c$
%if $c\in C_\theta$,
%or zero, otherwise.
%Note that $\delta(\theta) = \max_{c\in \Plane} \delta(\theta, c)$.

\begin{figure}[tb]
\begin{center}
\includegraphics[width=.65\textwidth]{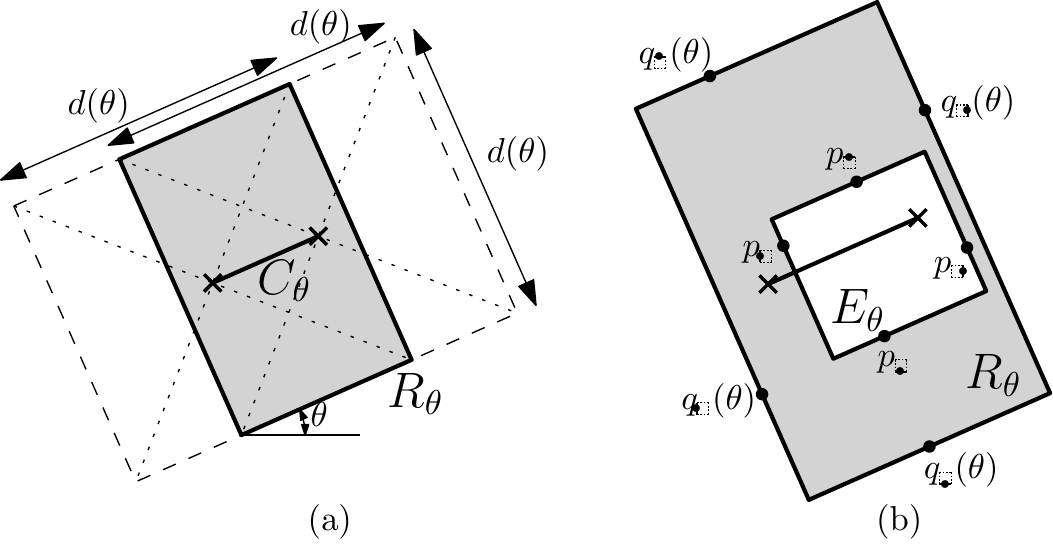}
\end{center}
%\vspace{-5mm}
\caption{ (a) How $C_\theta$ is formed,
 when it is $\theta$-aligned.
 Two endpoints of $C_\theta$ are the centers of the leftmost and rightmost
 smallest $\theta$-aligned square enclosing $P$ and thus enclosing $R_\theta$.
(b) For a MER class $E$,
 $E_\theta$ is a maximal empty $\theta$-aligned rectangle contained in $R_\theta$.}
% \vspace{-5mm}
\label{fig:C_delta_E}
\end{figure}

Let $C_\theta$ be the set of centers of all smallest $\theta$-aligned squares
enclosing $P$.
Observe that $C_\theta$ forms a line segment whose orientation is
$\theta$ if $d(\theta) = w_{\theta+\pi/2}(\extr_\mt(\theta), \extr_\mb(\theta))$,
or
$\theta+\pi/2$ if $d(\theta) = w_\theta(\extr_\ml(\theta), \extr_\mr(\theta))$.
Further, its length $|C_\theta|$ is equal to
$|w_{\theta+\pi/2}(\extr_\mt(\theta), \extr_\mb(\theta))
 - w_\theta(\extr_\ml(\theta), \extr_\mr(\theta))|$.
See \figurename~\ref{fig:C_delta_E}(a)
for an illustration when $d(\theta) = w_{\theta+\pi/2}(\extr_\mt(\theta), \extr_\mb(\theta))$, so $C_\theta$ is a $\theta$-aligned segment.
Lemma~\ref{lem:fixed_orient_minarea_sa} implies that
any minimum-area $\theta$-aligned square annulus is centered at
some $c \in C_\theta$, and so is its inner square.

Consider any MER class $E = (p_\mt, p_\mb, p_\ml, p_\mr; J)$.
Define $\delta_E(\theta, c)$ for $\theta \in [0, \pi/2)$ and $c \in C_\theta$
to be the side length of the largest $\theta$-aligned square centered at $c$
contained in $E_\theta$ if $\theta \in J$ and $c \in C_\theta$,
or zero, otherwise.
Specifically, it holds that
\[ \delta_E(\theta, c) :=
    \begin{cases}
      \min\{w_{\theta+\pi/2}(p_\mt, c), w_{\theta+\pi/2}(p_\mb, c),
       w_{\theta}(p_\ml, c), w_{\theta}(p_\mr, c)\}
          & \text{if $\theta \in J$ and $c \in C_\theta\cap E_\theta$}\\
      0 & \text{otherwise}
    \end{cases}.
\]
See \figurename~\ref{fig:C_delta_E}(b) for an illustration.
We then define $\delta_E(\theta) := \max_{c\in C_\theta} \delta_E(\theta, c)$.

We then consider a primary interval $I \subset [0,\pi/2)$
and fixes the four extreme points $\extr_\ma = \extr_\ma(\theta)$
for $\ma \in \{\mt, \mb, \ml, \mr\}$.
Note that the functions $\delta_E(\theta, c)$ and $\delta_E(\theta)$
in $I \cap J$
are dependent only on the eight points
$\extr_\mt, \extr_\mb, \extr_\ml, \extr_\mr, p_\mt, p_\mb, p_\ml, p_\mr \in P$,
so are of constant complexity.
More precisely, we observe the following, while
its proof is postponed to Appendix A.\@
\newcounter{saveLemDelta}
\addtocounter{saveLemDelta}{\value{lemma}}
\begin{lemma} \label{lem:delta}
 For any MER class $E$ with valid interval $J$ and any primary interval $I$,
 the function $\delta_E(\theta)$ over $\theta \in I\cap J$
 is piecewise base-$0$ frequency-$1$ sinusoidal with $O(1)$ breakpoints.
\end{lemma}

Now, let $a_E(\theta) := (d(\theta))^2 - (\delta_E(\theta))^2$
be the minimum possible area of a square annulus enclosing $P$
whose outer square is a smallest $\theta$-aligned square enclosing $P$
and inner square is contained in $E_\theta$.
Observe that $a(\theta) = \min_E a_E(\theta)$
since the inner square of a minimum-area $\theta$-aligned square annulus
should be contained in a maximal empty $\theta$-aligned rectangle.
In the minimum-area square annulus problem in arbitrary orientation,
we want to compute $\min_{\theta \in [0, \pi/2)} a(\theta)$ and this is
equivalent to compute
 \[      \min_{\theta \in [0, \pi/2)} \min_E a_E(\theta) =
     \min_E \min_{\theta \in [0, \pi/2)} a_E(\theta).\]
This indeed describes our algorithm.

Our algorithm runs as follows:
First, as done for the rectangular case,
we compute all primary intervals and MER classes in $O(n^3)$ time
using the algorithm by Chaudhuri et al.~\cite{cnd-lerps-03}.
We also specify and store all primary intervals $I$ relevant to each MER class
in time $O(n^3)$ by Lemma~\ref{lem:relevant_primary_int},
as done in the previous section.
For each MER class $E=(p_\mt, p_\mb, p_\ml, p_\mr; J)$
and each primary interval $I$ relevant to $E$,
we minimize $a_E(\theta) = (d(\theta))^2 - (\delta_E(\theta))^2$
over $\theta \in I \cap J$.
As discussed above and by Lemma~\ref{lem:delta},
both functions $d(\theta)$ and $\delta_E(\theta)$
are piecewise base-$0$ frequency-$1$ sinusoidal with $O(1)$ breakpoints
in $I\cap J$.
Thus, the function $a_E$ is piecewise frequency-$2$ sinusoidal
with $O(1)$ breakpoints by Observations~\ref{obs:sinusoidal_sum}
and~\ref{obs:sinusoidal_product},
so can be minimized in $O(1)$ time.
By taking the minimum over those $O(n^3)$ minima we have computed
in the above process,
we obtain the minimum possible area of any square annulus enclosing $P$
and the corresponding annulus.
Therefore, we conclude the following theorem.
\begin{theorem} \label{thm:A-S-a}
 A minimum-area square annulus enclosing $n$ points over all orientations
 can be computed in $O(n^3)$ time.
\end{theorem}

The same approach also applies to the minimum-width problem.
We define $w_E(\theta) := d(\theta)/2 - \delta_E(\theta)/2$,
which is the minimum possible width of a square annulus enclosing $P$
whose outer square is a smallest $\theta$-aligned square enclosing $P$
and inner square is contained in $E_\theta$.
We then have $w(\theta) = \min_E w_E(\theta)$
and want to compute
\[ \min_{\theta \in [0, \pi/2)} w(\theta)
 = \min_{\theta \in [0, \pi/2)} \min_E w_E(\theta) =
     \min_E \min_{\theta \in [0, \pi/2)} w_E(\theta).\]
Thus, to compute the minimum possible width and a corresponding square annulus,
we minimize $w_E(\theta) = d(\theta)/2 - \delta_E(\theta)/2$
over $\theta \in I \cap J$,
for each MER class $E=(p_\mt, p_\mb, p_\ml, p_\mr; J)$
and each primary interval $I$ relevant to $E$.
Again by Lemma~\ref{lem:delta},
we deduce that $w_E(\theta)$ in $I\cap J$ is a piecewise sinusoidal function
of frequency $1$ and base $0$
with $O(1)$ breakpoints, so can be minimized in $O(1)$ time.

Hence, we also conclude the following.
\begin{theorem} \label{thm:W-S-a}
 A minimum-width square annulus enclosing $n$ points over all orientations
 can be computed in $O(n^3)$ time.
\end{theorem}

Note that this even improves the previously fastest algorithm
of running time $O(n^3 \log n)$ by the author~\cite{b-cmwsaao-18}.

Both algorithms for the minimum-width and minimum-area problems
can indeed find \emph{all} minimum-width or minimum-area square annuli
enclosing $P$ over all orientations in the same time bound.
Therefore, we conclude the following corollary.
\begin{corollary} \label{coro:WA/AW-S-a}
 A minimum-width minimum-area or minimum-area minimum-width square annulus
 enclosing $n$ points over all orientations can be computed in $O(n^3)$ time.
\end{corollary}

%%%%%%%%%%%%%%%%%%%%%%%%%%%%%%%%%%%%%%%%%%%%%%%%%%%%%%
\section{Concluding Remarks} \label{sec:conclusion}
%%%%%%%%%%%%%%%%%%%%%%%%%%%%%%%%%%%%%%%%%%%%%%%%%%%%%
We have presented first algorithms on the minimum-area
square and rectangular annulus problem.
On one hand, the minimum-area problem in a fixed orientation is shown
to have a close relation
with the minimum-width problem and the largest empty rectangle problem,
so can be solved by existing algorithms.
On the other hand, our algorithms for the problem in arbitrary orientation
make use of the concept of maximal empty rectangles among input points.
This approach also applies to the problem of computing
the \emph{largest empty square} among $P$ over all orientations.
An \emph{empty square} among $P$ is an empty rectangle among $P$ that is a square.
\begin{theorem} \label{thm:largest_empty_square}
 Given a set $P$ of $n$ points, a largest empty square among $P$
 over all orientations can be computed in $O(n^3)$ time.
\end{theorem}
\begin{proof}
Consider any MER class $E = (p_\mt, p_\mb, p_\ml, p_\mr; J)$.
Let $s_E(\theta)$ be the side length of a largest empty
$\theta$-aligned square contained in $E_\theta$.
Also, let $s(\theta)$ be the side length of a largest empty
$\theta$-aligned square among $P$.
Then, we have $s(\theta) = \max_E s_E(\theta)$.
Since we want to maximize $s(\theta)$ over $\theta \in [0, \pi/2)$ and
we have
\[ \max_{\theta \in [0, \pi/2)} s(\theta)
   = \max_{\theta \in [0, \pi/2)} \max_E  s_E(\theta)
   = \max_E \max_{\theta \in [0, \pi/2)} s_E(\theta),\]
this theorem is obtained by maximizing $s_E(\theta)$
for each MER class $E$ among the $n$ input points,
using the same algorithmic framework as in Theorems~\ref{thm:A-R-a}
and~\ref{thm:A-S-a}.

The last task to be done is to check $s_E(\theta)$ is
piecewise sinusoidal of frequency $1$ and base $0$ with $O(1)$ breakpoints.
This directly follows from the observation that
\[ s_E(\theta) =
 \min\{w_{\theta+\pi/2}(p_\mt, p_\mb), w_\theta(p_\ml, p_\mr)\}\]
since it depends only on the four fixed points $p_\mt$, $p_\mb$, $p_\ml$,
and $p_\mr$.
\end{proof}

It is quite surprising that this is the first nontrivial bound
for the problem in the literature, to our best knowledge.
The problem of computing a largest empty square in a fixed orientation
can be found in $O(n\log n)$ time
by searching the $L_\infty$ Voronoi diagram of the input points,
as also pointed out in early papers~\cite{as-facler-87,cdl-cler-86,nlh-merp-84}.

Our algorithmic results seem to have room for improvement
for several cases, in particular, for arbitrary orientation cases.
Note that the minimum-area square and (general) rectangular annulus problem
in a fixed orientation has a lower bound of $\Omega(n \log n)$
by our reductions shown in Lemmas~\ref{lem:A-S-f_W-S-f}
and~\ref{lem:A-R-f-LER}.
In particular, improving the algorithms for the minimum-area rectangular problem
automatically improves those for the largest empty rectangle problem.
For the problems in arbitrary orientation,
there is no known nontrivial lower bound, other than $\Omega(n \log n)$
or $\Omega(n)$
which follows directly from that for the fixed-orientation problem.

Our algorithms for square and rectangular annuli in arbitrary orientation
are heavily dependent on the maximal empty rectangles $n$ points.
It would be interesting to ask a different approach that avoids
computing all maximal empty rectangles, which already takes $O(n^3)$ time.
In particular, to compute a minimum-area square annulus or a largest empty square,
it is sufficient to check all \emph{maximal empty squares} among $P$,
instead of maximal empty rectangles.
Similarly to maximal empty rectangles,
a maximal empty square can be defined to be an empty square
such that its three sides contain a point in $P$.
For a fixed orientation $\theta \in [0,\pi/2)$, the number of
maximal empty squares is bounded by $O(n)$ since
each of them defines a vertex of the $L_\infty$ Voronoi diagram.
However, we are unaware of any result on the number of maximal empty squares
over all orientations.
Its upper bound is $O(n^3)$ from the number of maximal empty rectangles,
and it can be $\Omega(n^2)$ sometimes by a simple construction of $P$.
If the correct bound is subcubic, then
there would be hope to improve the $O(n^3)$ algorithms for
the minimum-area square annulus and the largest empty square problems
to subcubic-time algorithms.

%=====================end of document=================

%

{
\bibliographystyle{abbrv}
\bibliography{annuli}
}

%\newpage
\appendix
\section{Proof of Lemma~\ref{lem:delta}}
In this Appendix, we give a proof of Lemma~\ref{lem:delta}.
\newcounter{saveLemCounter}
\addtocounter{saveLemCounter}{\arabic{lemma}}
\setcounter{lemma}{\arabic{saveLemDelta}}
\begin{lemma}
 For any MER class $E$ with valid interval $J$ and any primary interval $I$,
 the function $\delta_E(\theta)$ over $\theta \in I\cap J$
 is piecewise base-$0$ frequency-$1$ sinusoidal with $O(1)$ breakpoints.
\end{lemma}
\setcounter{lemma}{\arabic{saveLemCounter}}

We first recall some definitions:
$R_\theta$ is the smallest $\theta$-aligned enclosing rectangle for $P$;
$C_\theta$ is the set of centers of all smallest $\theta$-aligned enclosing
squares for $P$;
$d(\theta) = \max\{ w_{\theta+\pi/2}(\extr_\mt(\theta), \extr_\mb(\theta)),$
 $w_\theta(\extr_\ml(\theta), \extr_\mr(\theta))\}$
is the side length of a smallest $\theta$-aligned enclosing square for $P$;
For each MER class $E = (p_\mt, p_\mb, p_\ml, p_\mr; J)$,
and for any $\theta \in [0, \pi/2)$ and any $c \in C_\theta$
\[ \delta_E(\theta, c) =
    \begin{cases}
      \min\{w_{\theta+\pi/2}(p_\mt, c), w_{\theta+\pi/2}(p_\mb, c),
       w_{\theta}(p_\ml, c), w_{\theta}(p_\mr, c)\}
          & \theta \in J, c \in C_\theta\cap E_\theta\\
      0 & \text{otherwise}
    \end{cases}
\]
is the side length of the largest $\theta$-aligned square centered at $c$
contained in $E_\theta$ if $\theta \in J$ and $c \in C_\theta$,
or zero, otherwise;
lastly, $\delta_E(\theta) = \max_{c\in C_\theta} \delta_E(\theta, c)$.
%$R_\theta$ is the smallest $\theta$-aligned enclosing rectangle for $P$,
%$C_\theta$ is the set of centers of all smallest $\theta$-aligned enclosing
%squares for $P$.
%$d(\theta) = \max\{ w_{\theta+\pi/2}(\extr_\mt(\theta), \extr_\mb(\theta)),$
% $w_\theta(\extr_\ml(\theta), \extr_\mr(\theta))\}$
%is the side length of a smallest $\theta$-aligned enclosing square for $P$.
%For any MER class $E = (p_\mt, p_\mb, p_\ml, p_\mr; J)$,
%and for $\theta \in [0, \pi/2)$ and $c \in C_\theta$
%\[ \delta_E(\theta, c) =
%    \begin{cases}
%      \min\{w_{\theta+\pi/2}(p_\mt, c), w_{\theta+\pi/2}(p_\mb, c),
%       w_{\theta}(p_\ml, c), w_{\theta}(p_\mr, c)\}
%          & \text{if $\theta \in J$ and $c \in C_\theta\cap E_\theta$}\\
%      0 & \text{otherwise}
%    \end{cases}
%\]
%is the side length of the largest $\theta$-aligned square centered at $c$
%contained in $E_\theta$ if $\theta \in J$ and $c \in C_\theta$,
%or zero, otherwise.
%Lastly, $\delta_E(\theta) = \max_{c\in C_\theta} \delta_E(\theta, c)$.

Here, for convenience, we refine primary intervals as follows.
Throughout this proof, a \emph{primary interval} is redefined to be
a maximal interval $I \subset [0, \pi/2)$ such that
the four extreme points are fixed,
so $\extr_\mt = \extr_\mt(\theta)$, $\extr_\mb = \extr_\mb(\theta)$,
$\extr_\ml = \extr_\ml(\theta)$, and $\extr_\mr = \extr_\mr(\theta)$
over all $\theta \in I$,
and further there is no $\theta \in I$ with
$w_{\theta+\pi/2}(\extr_\mt, \extr_\mb) = w_\theta(\extr_\ml, \extr_\mr)$.
The second condition keeps the function $d(\theta)$ to be in a simpler form
in a primary interval:
$d(\theta) = w_{\theta+\pi/2}(\extr_\mt, \extr_\mb)$
or $d(\theta) = w_\theta(\extr_\ml, \extr_\mr)$.
For a fixed tuple of $(\extr_\mt, \extr_\mb, \extr_\ml, \extr_\mr)$,
there is at most one $\theta$ such that
$w_{\theta+\pi/2}(\extr_\mt, \extr_\mb) = w_\theta(\extr_\ml, \extr_\mr)$
by Observation~\ref{obs:sinusoidal_sum}.
Hence, in order to prove the lemma,
it suffices to consider each of these refined primary interval.

Now, we give a proof for the lemma.
Consider any MER class $E = (p_\mt, p_\mb, p_\ml, p_\mr; J)$
and a (refined) primary interval $I$.
If $I \cap J = \emptyset$, then the lemma follows trivially,
so suppose in the following that $I \cap J \neq \emptyset$.
Let $\extr_\mt = \extr_\mt(\theta)$,
$\extr_\mb = \extr_\mb(\theta)$, $\extr_\ml = \extr_\ml(\theta)$,
and $\extr_\mr = \extr_\mr(\theta)$ for any $\theta \in I$.
Without loss of generality, we assume that
$w_{\theta+\pi/2}(\extr_\mt, \extr_\mb) \geq w_\theta(\extr_\ml, \extr_\mr)$
for any $\theta \in I$,
so $d(\theta) = w_{\theta+\pi/2}(\extr_\mt, \extr_\mb)$.
Then, $C_\theta$ forms a $\theta$-aligned line segment
whose midpoint is located at the center of $R_\theta$,
that is, the intersection of two diagonals of $R_\theta$,
and whose length is $|C_\theta| = w_{\theta+\pi/2}(\extr_\mt, \extr_\mb) - w_\theta(\extr_\ml, \extr_\mr)$.

\begin{figure}[tb]
\begin{center}
\includegraphics[width=.7\textwidth]{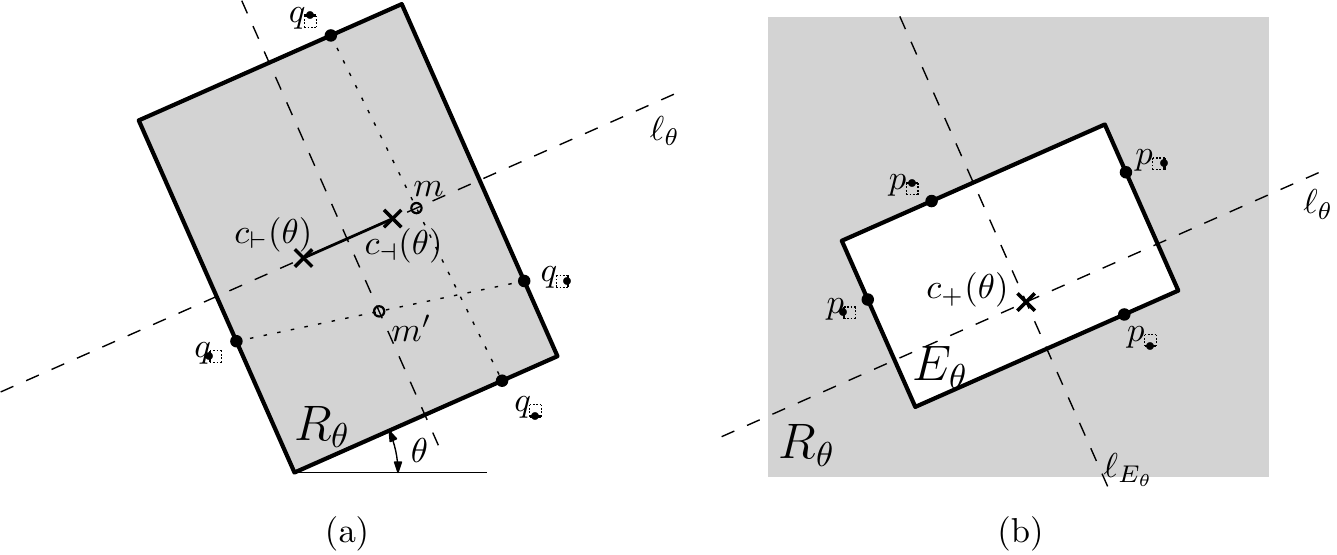}
\end{center}
%\vspace{-5mm}
\caption{ How (a) $c_\clend(\theta)$ and $c_\crend(\theta)$, and (b) $c_\cmid(\theta)$ are determined. $\ell_{E_\theta}$ is the $(\theta+\pi/2)$-aligned line halving $E_\theta$.}
% \vspace{-5mm}
\label{fig:m_cl_cr}
\end{figure}

Let $\ell_\theta$ be the line extending $C_\theta$.
We let $m$ be the midpoint of line segment $\seg{\extr_\mt \extr_\mb}$
and $m'$ be the midpoint of line segment $\seg{\extr_\ml \extr_\mr}$.
Note that the $\theta$-aligned line through $m$ horizontally halves $R_\theta$
and the $(\theta+\pi/2)$-aligned line through $m'$ vertically halves $R_\theta$.
Observe that $\ell_\theta$ is the $\theta$-aligned line through $m$
for any $\theta \in I$,
so $\ell_\theta$ rotates around $m$ as $\theta \in I$ continuously increases.
Also, we denote by $c_\clend(\theta)$ and $c_\crend(\theta)$
the left and right endpoints of $C_\theta$ for $\theta \in I$,
and let $c_\cmid(\theta) \in \ell_\theta$ be the intersection point
between $\ell_\theta$ and the $(\theta+\pi/2)$-aligned line
that halves $R_\theta$.
See \figurename~\ref{fig:m_cl_cr}.

We are interested in three empty squares
centered at $c_\clend(\theta)$, $c_\crend(\theta)$, and $c_\cmid(\theta)$,
respectively,
that are contained in $E_\theta$.
Define
\[
\rho_\clend(\theta)  := \delta_E(\theta, c_\clend(\theta)), \quad
\rho_\crend(\theta)  := \delta_E(\theta, c_\crend(\theta)), \quad \text{and} \quad
\rho_\cmid(\theta) := \delta_E(\theta, c_\cmid(\theta))
\]
to be the side lengths of the largest $\theta$-aligned squares
contained in $E_\theta$ centered at
$c_\clend(\theta)$, $c_\crend(\theta)$, and $c_\cmid(\theta)$,
respectively.
We claim that these three functions over $\theta \in I \cap J$
are all piecewise frequency-$1$ base-$0$ sinusoidal.
For the purpose, we need more observations.
For any $p\in \Plane$, let $\lambda_\theta(p)$ be the signed distance
from $m$ to the orthogonal projection of $p$ onto $\ell_\theta$
so that it is positive if $p$ is to the right of $m$
or negative if $p$ is to the left of $m$.
Specifically, we have
 \[ \lambda_\theta(p) = |\seg{pm}|\cdot \cos(\theta_{pm} - \theta), \]
where $\theta_{pm} \in [0,\pi)$ is the orientation of segment $\seg{pm}$.
Note that $w_\theta(p, m) = |\lambda_\theta(p)|$.

\begin{figure}[tb]
\begin{center}
\includegraphics[width=.6\textwidth]{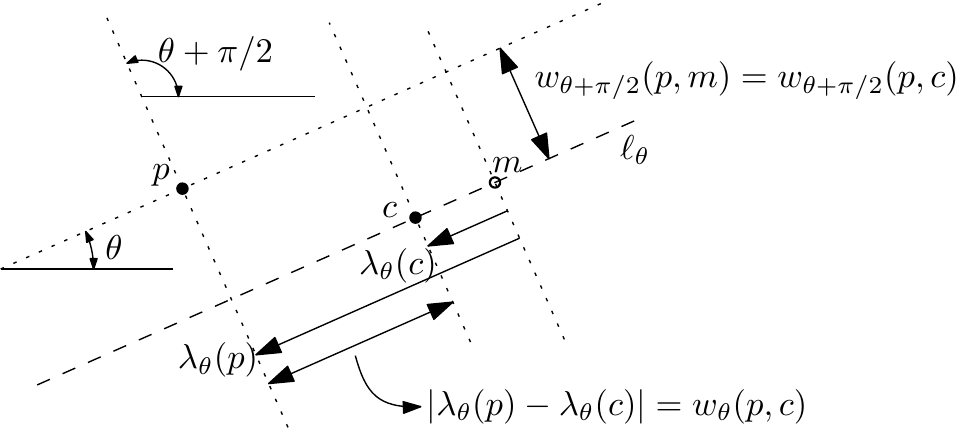}
\end{center}
%\vspace{-5mm}
\caption{Illustration to the proof of Observation~\ref{obs:lambda}.}
% \vspace{-5mm}
\label{fig:obs_lambda}
\end{figure}

We then observe the following.
\begin{observation} \label{obs:lambda}
 For any $\theta \in I \cap J$,
  any $c \in \ell_\theta$ and any $p\in \Plane$, it holds that
 \begin{itemize} \denseitems
  \item $w_{\theta+\pi/2}(p, c) = w_{\theta+\pi/2}(p, m)$, and
  \item $w_\theta(p, c) = |\lambda_\theta(p) - \lambda_\theta(c)|$.
 \end{itemize}
\end{observation}
\begin{proof}
See \figurename~\ref{fig:obs_lambda} for an illustration.
Since $c \in \ell_\theta$,
the distance between the orthogonal projections of $c$ and any $p\in \Plane$ onto
a $(\theta+\pi/2)$-aligned line is equal to that between
the projections of $m$ and $p$.
This proves the first equation.

The value of $w_\theta(p, c)$ is equal to the distance between
the orthogonal projections of $p$ and $c$ onto $\ell_\theta$, and hence
is represented by the difference of their relative distances to $m$,
which is $|\lambda_\theta(p) - \lambda_\theta(c)|$.
\end{proof}

Applying the above observation to the case where $c = c_\clend(\theta)$
or $c=c_\crend(\theta)$, we obtain the following.
\begin{observation} \label{obs:lambda_c_l}
 For any $\theta \in I\cap J$ and any $p\in \Plane$,
 it holds that
 \begin{itemize} \denseitems
  \item $w_{\theta+\pi/2}(p, c_\clend(\theta)) = w_{\theta+\pi/2}(p, c_\crend(\theta)) = w_{\theta+\pi/2}(p, m)$,
 % \item $w_{\theta+\pi/2}(p, c_\crend(\theta)) = w_{\theta+\pi/2}(p, m)$,
  \item $w_\theta(p, c_\clend(\theta)) =
   |\lambda_\theta(p) - (\lambda_\theta(m')-\frac{1}{2}(w_{\theta+\pi/2}(\extr_\mt, \extr_\mb) - w_\theta(\extr_\ml, \extr_\mr)))|$, and
  \item $w_\theta(p, c_\crend(\theta)) =
   |\lambda_\theta(p) - (\lambda_\theta(m')+\frac{1}{2}(w_{\theta+\pi/2}(\extr_\mt, \extr_\mb) - w_\theta(\extr_\ml, \extr_\mr)))|$.
 \end{itemize}
\end{observation}
\begin{proof}
The first equation directly follows from Observation~\ref{obs:lambda}
when $c = c_\clend(\theta)$ and $c = c_\crend(\theta)$.

To see the last two, recall that $C_\theta \subset \ell_\theta$ is a line segment
of length $w_{\theta+\pi/2}(\extr_\mt, \extr_\mb) - w_\theta(\extr_\ml, \extr_\mr)$,
and its midpoint is the intersection of $\ell_\theta$
and the $(\theta+\pi/2)$-aligned line through $m'$.
Thus, we have
\[
 \lambda_\theta(c_\clend(\theta)) = \lambda_\theta(m') - |C_\theta|/2,
 \quad \text{and} \quad
 \lambda_\theta(c_\crend(\theta)) = \lambda_\theta(m') + |C_\theta|/2,
\]
where $|C_\theta|=w_{\theta+\pi/2}(\extr_\mt, \extr_\mb) - w_\theta(\extr_\ml, \extr_\mr)$ is the length of segment $C_\theta$.
Plugging these into Observation~\ref{obs:lambda}, we obtain the last two equations.
\end{proof}

Similarly, for the case $c = c_\cmid(\theta)$, we observe the following.
\begin{observation} \label{obs:lambda_c_mid}
 For any $\theta \in I\cap J$ and any $p\in \Plane$,
 it holds that
 \begin{itemize}\denseitems
  \item $w_{\theta+\pi/2}(p, c_\cmid(\theta)) = w_{\theta+\pi/2}(p, m)$, and
  \item $w_\theta(p, c_\cmid(\theta)) =
   |\lambda_\theta(p) - \frac{1}{2}(\lambda_\theta(p_\mr)+\lambda_\theta(p_\ml))|$.
 \end{itemize}
\end{observation}
\begin{proof}
This follows from Observation~\ref{obs:lambda}
and the fact that
$\lambda_\theta(c_\cmid(\theta)) = \frac{1}{2}(\lambda_\theta(p_\mr)+\lambda_\theta(p_\ml))$.
\end{proof}

From Observations~\ref{obs:lambda_c_l} and~\ref{obs:lambda_c_mid},
it is easy to see that the three functions $\rho_\clend$, $\rho_\crend$, and $\rho_\cmid$
behave nicely.
\begin{lemma} \label{lem:rho}
 The three functions $\rho_\clend(\theta)$, $\rho_\crend(\theta)$, and $\rho_\cmid(\theta)$ over $\theta \in I \cap J$
 are piecewise frequency-$1$ base-$0$ sinusoidal with $O(1)$ breakpoints.
\end{lemma}
\begin{proof}
First, consider the function $\rho_\clend(\theta) = \delta_E(\theta, c_\clend(\theta))$.
By definition, we have
\[ \rho_\clend(\theta) = \min\{w_{\theta+\pi/2}(p_\mt, c_\clend(\theta)),
      w_{\theta+\pi/2}(p_\mb, c_\clend(\theta)),
       w_{\theta}(p_\ml, c_\clend(\theta)), w_{\theta}(p_\mr, c_\clend(\theta))\}, \]
if $c_\clend(\theta) \in E_\theta$,
and $\rho_\clend(\theta) = 0$, otherwise.
In case of $c_\clend(\theta) \in E_\theta$,
Observation~\ref{obs:lambda_c_l} implies that
\begin{align*}
 \rho_\clend(\theta) = \min\{w_{\theta+\pi/2}(p_\mt, m), & ~ w_{\theta+\pi/2}(p_\mb, m), \\
      & |\lambda_\theta(p_\ml) - (\lambda_\theta(m')-|C_\theta|/2)|, ~
      |\lambda_\theta(p_\mr) - (\lambda_\theta(m')-|C_\theta|/2)|\},
\end{align*}
where $|C_\theta|=w_{\theta+\pi/2}(\extr_\mt, \extr_\mb) - w_\theta(\extr_\ml, \extr_\mr)$ is the length of $C_\theta$.
Now, we check the four terms in the $\min\{ \cdot \}$.
The first two are obviously piecewise frequency-$1$ base-$0$ sinusoidal.
For the last two, observe that
$\lambda_\theta(p)$ for any fixed point $p\in \Plane$ is frequency-$1$ base-$0$
sinusoidal by definition,
and $|C_\theta|$ is also piecewise frequency-$1$ base-$0$ sinusoidal.
Thus, the last two terms are also piecewise frequency-$1$ base-$0$ sinusoidal
by Observation~\ref{obs:sinusoidal_sum}.
Hence, $\rho_\clend(\theta)$ is piecewise frequency-$1$ base-$0$ sinusoidal.
The number of breakpoints of $\rho_\clend$ is bounded by $O(1)$
since it depends only on a set of constantly many fixed points, namely,
$\{ \extr_\mt, \extr_\mb, \extr_\ml, \extr_\mr, p_\mt, p_\mb, p_\ml, p_\mr, m, m'\}$.
The proof for $\rho_\crend(\theta)$ is almost identical to above.

Now, consider the function $\rho_\cmid(\theta) = \delta_E(\theta, c_\cmid(\theta))$.
By definition, we have
\[ \rho_\cmid(\theta) = \min\{w_{\theta+\pi/2}(p_\mt, c_\cmid(\theta)),
      w_{\theta+\pi/2}(p_\mb, c_\cmid(\theta)),
       w_{\theta}(p_\ml, c_\cmid(\theta)), w_{\theta}(p_\mr, c_\cmid(\theta))\}, \]
if $c_\cmid(\theta) \in E_\theta \cap C_\theta$,
and $\rho_\cmid(\theta) = 0$, otherwise.
For the case where $c_\cmid(\theta) \in E_\theta \cap C_\theta$,
Observation~\ref{obs:lambda_c_mid} implies that
\[ \rho_\cmid(\theta) = \min\{w_{\theta+\pi/2}(p_\mt, m),w_{\theta+\pi/2}(p_\mb, m),
   \frac{1}{2}(\lambda_\theta(p_\mr)-\lambda_\theta(p_\ml)) \}. \]
By a similar argument as above,
$\rho_\cmid(\theta)$ is also piecewise frequency-$1$ base-$0$ sinusoidal
with $O(1)$ breakpoints.
\end{proof}

We then show that the maximum $\delta_E(\theta)$ of
$\delta_E(\theta, c)$ over $c\in E_\theta \cap C_\theta$,
for each $\theta \in I\cap J$,
is represented by the three functions $\rho_\clend$, $\rho_\crend$, and $\rho_\cmid$.
\begin{lemma} \label{lem:deltaG_rho}
 The function $\delta_E(\theta)$ over $\theta \in I\cap J$
 is represented as follows:
 \[
  \delta_E(\theta) = \begin{cases}
   0 & C_\theta \cap E_\theta = \emptyset \\
   \rho_\cmid(\theta)
    & c_\cmid(\theta) \in E_\theta \cap C_\theta \\
   \rho_\clend(\theta)
    & c_\cmid(\theta) \notin E_\theta \cap C_\theta,
      c_\clend(\theta) \in E_\theta, c_\crend(\theta) \notin E_\theta\\
   \rho_\crend(\theta)
    & c_\cmid(\theta) \notin E_\theta \cap C_\theta,
      c_\clend(\theta) \notin E_\theta, c_\crend(\theta) \in E_\theta\\
   \max\{\rho_\clend(\theta), \rho_\crend(\theta) \}
    & c_\cmid(\theta) \notin E_\theta \cap C_\theta, C_\theta \subset E_\theta
  \end{cases}.
 \]
\end{lemma}
\begin{proof}
Recall that $\delta_E(\theta)
= \max_{c\in C_\theta} \delta_E(\theta,c)$, and
$\delta_E(\theta, c) = \min\{w_{\theta+\pi/2}(p_\mt, c),$ $w_{\theta+\pi/2}(p_\mb, c),$ $w_\theta(p_\ml, c),$ $w_\theta(p_\mr, c)\}$ if
$c\in C_\theta\cap E_\theta$, or zero, otherwise.
For $c \in \ell_\theta \cap E_\theta$,
note that
\begin{itemize} \denseitems
 \item  $w_{\theta+\pi/2}(p_\mt, c) = w_{\theta+\pi/2}(p_\mt, m)$,
 \item $w_{\theta+\pi/2}(p_\mb, c) = w_{\theta+\pi/2}(p_\mb, m)$,
 \item $w_\theta(p_\ml, c) = \lambda_\theta(c) - \lambda_\theta(p_\ml)$, and
 \item $w_\theta(p_\mr, c) = \lambda_\theta(p_\mr) - \lambda_\theta(c)$,
\end{itemize}
by Observation~\ref{obs:lambda}.
Hence, for a fixed $\theta \in I \cap J$,
the first two terms are constant over all $c \in C_\theta \cap E_\theta$,
and the last two terms are linear functions of $\lambda_\theta(c)$
of slope $1$ and $-1$.
This implies that the function
$c \mapsto \min\{w_{\theta+\pi/2}(p_\mt, c), w_{\theta+\pi/2}(p_\mb, c), w_\theta(p_\ml, c), w_\theta(p_\mr, c)\}$ on $c \in \ell_\theta\cap E_\theta$
is piecewise linear with
three pieces whose slopes are $1$, $0$, and $-1$ in this order.
Observe that the maximum of this function over $c \in \ell_\theta\cap E_\theta$
is always attained at $c = c_\cmid(\theta)$
since
\[ w_\theta(p_\ml, c_\cmid(\theta)) = w_\theta(p_\mr, c_\cmid(\theta)) =
(\lambda_\theta(p_\mr) - \lambda_\theta(p_\ml))/2 \geq
 \min\{w_\theta(p_\ml, c), w_\theta(p_\mr, c)\} \]
for any $c \in \ell_\theta \cap E_\theta$.

Now, consider the function $\delta_E(\theta, c)$ and
its maximum $\delta_E(\theta)$ over $c \in C_\theta \cap E_\theta$.
We handle the five cases separately.
For the first case where $C_\theta \cap E_\theta = \emptyset$
it is trivial that $\delta_E(\theta) = 0$.

Second, assume that $c_\cmid(\theta) \in E_\theta \cap C_\theta$.
Then, by the above discussion, we have
$\delta_E(\theta) = \delta_E(\theta, c_\cmid(\theta)) = \rho_\cmid(\theta)$.

Next, suppose that $C_\theta \cap E_\theta \neq \emptyset$
and $c_\cmid(\theta) \notin E_\theta \cap C_\theta$.
Then, again by the above discussion, we have
either $\delta_E(\theta) = \delta_E(\theta, c_\clend(\theta))$
or $\delta_E(\theta) = \delta_E(\theta, c_\crend(\theta))$.
On the other hand, at least one of $c_\clend(\theta)$ or $c_\crend(\theta)$
must be contained in $E_\theta$,
since $C_\theta$ is a $\theta$-aligned line segment and
$c_\cmid(\theta) \notin C_\theta$.
Thus,
\begin{itemize} \denseitems
 \item if $c_\crend(\theta)\notin E_\theta$,
then $\delta_E(\theta) = \delta_E(\theta, c_\clend(\theta)) = \rho_\clend(\theta)$;
 \item if $c_\clend(\theta)\notin E(\theta)$,
then $\delta_E(\theta) = \delta_E(\theta, c_\crend(\theta)) = \rho_\crend(\theta)$;
 \item otherwise, if both $c_\clend(\theta)$ and $c_\crend(\theta)$ lie in $E_\theta$,
then we have
\[\delta_E(\theta) = \max\{\delta_E(\theta, c_\clend(\theta)), \delta_E(\theta, c_\crend(\theta)\} = \max\{\rho_\clend(\theta), \rho_\crend(\theta))\}.\]
\end{itemize}
This completes the proof of the lemma.
\end{proof}

Note that each of $\rho_\clend(\theta)$, $\rho_\crend(\theta)$, and $\rho_\cmid(\theta)$
are piecewise sinusoidal of frequency $1$ and base $0$
as observed in Lemma~\ref{lem:rho}.
Hence, by Lemma~\ref{lem:deltaG_rho},
we verify that $\delta_E(\theta)$ is piecewise sinusoidal of frequency $1$
and base $0$, too.
The number of breakpoints of $\delta_E(\theta)$ over $\theta \in I \cap J$
is also bounded by a constant as follows:

The function $\delta_E(\theta)$ has the five cases as described in
Lemma~\ref{lem:deltaG_rho}.
As $\theta$ continuously increases in $I \cap J$,
we turn into a different case from one whenever one of the following
\emph{events} occurs:
\begin{enumerate}[(i)] \denseitems
 \item When $c_\clend(\theta)$ or $c_\crend(\theta)$ lies
 on the left or right side of $E_\theta$,
 equivalently, when either one of the following becomes zero:
 $w_\theta(p_\ml, c_\clend(\theta))$,
 $w_\theta(p_\ml, c_\crend(\theta))$, $w_\theta(p_\mr, c_\clend(\theta))$, and
 $w_\theta(p_\mr, c_\crend(\theta))$.
 \item When the top or bottom side of $E_\theta$ lies on $\ell_\theta$,
 equivalently, when either $w_{\theta+\pi/2}(p_\mt, m) = 0$ or
 $w_{\theta+\pi/2}(p_\mb, m) = 0$.
 \item When either $c_\clend(\theta) = c_\cmid(\theta)$ or $c_\crend(\theta) = c_\cmid(\theta)$,
 equivalently,
 when either $w_\theta(p_\ml, c_\clend(\theta)) = w_\theta(p_\mr, c_\clend(\theta))$
 or $w_\theta(p_\ml, c_\crend(\theta)) = w_\theta(p_\mr, c_\crend(\theta))$.
\end{enumerate}
As described above, each event occurs when some sinusoidal function of
frequency $1$ and base $0$ becomes zero,
so there are a constantly many number of such events
by Observation~\ref{obs:sinusoidal_sum}.

This completes the proof of Lemma~\ref{lem:delta}.
\hfill\copy\ProofSym
\end{document}